\documentclass[reprint,aps,twoside,onecolumn,superscriptaddress]{revtex4-2}

\RequirePackage[utf8]{inputenc}
\RequirePackage[english]{babel}
\RequirePackage[T1]{fontenc}
\RequirePackage{amsmath,amsthm,amssymb,amsfonts,amsbsy,stmaryrd,
    bm,colonequals,nccmath,mathtools,mathrsfs,xfrac}
\RequirePackage{dsfont,lipsum}
\RequirePackage{hyperref,comment,ragged2e}
\RequirePackage{braket,xifthen,makecell}
\RequirePackage{tikz}
\RequirePackage[mode=build]{standalone}
\RequirePackage{tcolorbox}
\RequirePackage{mdframed}
\RequirePackage{braket}
\RequirePackage{nicefrac}
\RequirePackage{tikz}
\RequirePackage{svg}
\RequirePackage[inline]{enumitem}

\usetikzlibrary{arrows,arrows.meta,3d,shapes,calc,decorations.pathreplacing,decorations.markings,positioning,intersections,shapes.symbols,positioning}

\DeclareMathAlphabet{\mymathbb}{U}{BOONDOX-ds}{m}{n}

\theoremstyle{plain}
\newtheorem{theorem}             {Theorem}
\newtheorem{proposition}[theorem]{Proposition}
\newtheorem{lemma}[theorem]{Lemma}

\newtheorem*{theorem*}    {Theorem}
\newtheorem*{proposition*}{Proposition}
\newtheorem*{lemma*}      {Lemma}
\newtheorem*{corollary*}  {Corollary}
\newtheorem*{conjecture*} {Conjecture}

\newtheorem{definition}[theorem]{Definition}

\theoremstyle{definition}

\newtheorem*{definition*}{Definition}
\newtheorem*{example*}   {Example}

\theoremstyle{remark}

\newtcolorbox{mybox}[2][]{
               = {yshift=-8pt},
  colback      = cyan!6!white,
  colframe     = cyan!1!black,
  halign       = flush left,
  fonttitle    = \bfseries\sffamily,
  colbacktitle = cyan!50!black,
  title        = #2,#1,
  }

\definecolor{dark-gray}{gray}{0.40}
\definecolor{quantumviolet}{HTML}{53257F} 
\definecolor{quantumlightviolet}{HTML}{A088B1}
\definecolor{quantumgreen}{HTML}{00826F}
\definecolor{quantumrose}{HTML}{EDB3FF} 
\definecolor{quantumdarkrose}{HTML}{F06292}
\definecolor{quantumturquoise}{HTML}{00C9AF}
\definecolor{quantumblue}{HTML}{85B1CC}
\definecolor{quantumdarkgray}{HTML}{4C4452}
\definecolor{quantumgray}{HTML}{555555}

\hypersetup{
    colorlinks=true,
    linkcolor=quantumgreen,
    citecolor=quantumviolet,
    urlcolor=quantumdarkrose,
    breaklinks=true,
}

\newcommand{\ie}{i.e.\ }
\newcommand{\eg}{e.g.\ }

\newcommand{\phsp}[2]{%
\ifthenelse{\isempty{#1}}
	{\rule[-.5\baselineskip]{0pt}{.5\baselineskip}}%
	{\rule[#1\baselineskip]{0pt}{#2\baselineskip}}%
}

\newcommand{\dummy}{\text{--}}
\newcommand{\Id}{\mathds{1}}
\newcommand{\Tr}{\text{Tr}}

\newcommand{\Mcomma}{\text{ ,}}

\newcommand{\defeq}{\colonequals}
\newcommand{\tuple}[1]{\left\langle #1 \right\rangle} 
\newcommand{\enset}   [1]{\mathopen{ \{ }#1\mathclose{ \} }} 
\newcommand{\family}   [1]{\mathopen{ ( }#1\mathclose{ ) }} 
\newcommand{\fdec}    [3]{#1 \colon #2 \longrightarrow #3}

\newcommand{\bfLambda}{\bm{\Lambda}}

\newcommand{\dd}{\text{d}}
\newcommand{\intg}[3]{\int_{#1}#2\,\mathrm{d}#3}

\newcommand{\Rext}{\overline{\mathbb{R}}}

\makeatletter
\newcommand{\leqnomode}{\tagsleft@true\let\veqno\@@leqno}
\newcommand{\reqnomode}{\tagsleft@false\let\veqno\@@eqno}
\makeatother

\makeatletter
\newcommand{\proglabel}[2]{%
   \protected@write \@auxout {}{\string \newlabel {#1}{{#2}{\thepage}{#2}{#1}{}} }%
   \hypertarget{#1}{#2}
}
\makeatother

\newcommand{\Bc}{\mathcal{B}}

\newcommand{\Fc}{\mathcal{F}}

\newcommand{\Mc}{\mathcal{M}}
\newcommand{\Oc}{\mathcal{O}}

\newcommand{\Pc}{\mathcal{P}}

\newcommand{\Xc}{\mathcal{X}}

\newcommand{\C}{\mathbb{C}}
\newcommand{\N}{\mathbb{N}}

\newcommand{\R}{\mathbb{R}}
\newcommand{\Z}{\mathbb{Z}}
\newcommand{\orcid}[1]{\href[hidelinks]{https://orcid.org/#1}{\includesvg[height=2ex]{ORCIDiD_iconvector}}}

\allowdisplaybreaks

\begin{document}

    \title{Contextuality and Wigner negativity are equivalent for continuous-variable quantum measurements}

    \author{Robert I. Booth}
    \email{robert.booth@ed.ac.uk}
    \affiliation{School of Informatics, University of Edinburgh, Edinburgh, United Kingdom}
    \affiliation{LORIA CNRS, Inria-MOCQUA, Universit\'e de Lorraine, F-54000 Nancy, France}
    \affiliation{Sorbonne Université, CNRS, LIP6, F-75005 Paris, France}
    \author{Ulysse Chabaud}
    \email{uchabaud@caltech.edu}
    \affiliation{Institute for Quantum Information and Matter, Caltech, Pasadena, CA 91125 USA}
    \author{Pierre-Emmanuel Emeriau}
    \email{pe.emeriau@quandela.com}
    \affiliation{Quandela, 10  Boulevard  Thomas  Gobert, 91120,  Palaiseau, France}
    %



\begin{abstract}
 Quantum computers promise considerable speed-ups with respect to their classical counterparts. However, the identification of the innately quantum features that enable these speed-ups is challenging. In the continuous-variable setting---a promising paradigm for the realisation of universal, scalable, and fault-tolerant quantum computing---contextuality and Wigner negativity have been perceived as two such distinct resources. Here we show that they are in fact equivalent for the standard models of continuous-variable quantum computing. While our results provide a unifying picture of continuous-variable resources for quantum speed-up, they also pave the way towards practical demonstrations of continuous-variable contextuality, and shed light on the significance of negative probabilities in phase-space descriptions of quantum mechanics.
\end{abstract}

\maketitle



\section{Introduction}

With the onset of quantum information theory, the weirdness of quantum mechanics has transitioned from being a bug to being a feature, and the first demonstrations of quantum speedup have recently been achieved~\cite{arute2019quantum,zhong2020quantum}, building on inherently nonclassical properties of physical systems.
While entanglement is used daily for the calibration of current quantum experiments, it was originally perceived as a `spooky action at a distance' by Einstein. This led him, Podolsky and Rosen (EPR) to speculate about the incompleteness of quantum mechanics~\cite{einsteincan1935} and the existence of a deeper theory over `hidden' variables reproducing the predictions of quantum mechanics without its puzzling nonlocal aspects. 

During the same period, Wigner was also looking for a more intuitive description of quantum mechanics, and he obtained a phase-space description akin to that of classical theory~\cite{Wigner1932}. However, a major difference with the classical case was that the Wigner function---the quantum equivalent of a classical probability distribution over phase space---could display negative values. These `negative probabilities' seemingly prevented a classical phase-space interpretation of quantum mechanics.

More than thirty years later, the seminal results of Bell \cite{belleinstein1964,bell1966} and Kochen and Specker \cite{kochen1975problem} ruled out the possibility of finding the underlying hidden-variable model for quantum mechanics envisioned by EPR, thus establishing nonlocality, and its generalisation contextuality, as fundamental properties of quantum systems.

At an intuitive level, contextuality and negativity of the Wigner function are
properties of quantum states that seek to capture similar characteristics of
quantum theory: the non-existence of a classical probability distribution that
describes the outcomes of the measurements of a quantum system. 

In more operational terms, contextuality is present whenever any hidden-variable description of the behaviour of a system is inconsistent with the basic assumptions that
\begin{enumerate*}[label=(\roman*)]
\item
all of its observable properties may be assigned definite values at all times, and
\item
jointly measuring compatible observables does not disturb these global value assignments, or, in other words, these assignments are context-independent.
\end{enumerate*} 
Aside from its foundational importance, contextuality has been increasingly identified as an essential ingredient for enabling a range of quantum-over-classical advantages in information processing tasks, which include the onset of universal quantum computing in certain computational models \cite{raussendorf2013contextuality,howard2014contextuality,abramsky2017contextual,bermejo2017contextuality,abramsky2017quantum}. 

Similarly, the negativity of the Wigner function, or Wigner negativity for short, is also crucial for quantum computational speedup as quantum computations described by nonnegative Wigner functions can be simulated efficiently classically~\cite{Mari2012}.

Importantly, quantum information can be encoded with discrete but also continuous variables (CV)~\cite{lloyd1999quantum}, using continuous quantum degrees of freedom of physical systems such as position or momentum.
The study of contextuality has mostly focused on the simpler discrete-variable
setting~\cite{abramsky2011sheaf,csw,Spekkens2005,Xiang2013,Helmut2006,bartosik2009experimental,kirchmair2009state}. In \cite{spekkens2008negativity}, it was shown that generalised contextuality is equivalent to the non-existence of a nonnegative quasiprobability representation. The caveat is that to check if a system is indeed contextual, one would have to consider all possible quasi-probability distributions.
Focusing on one particular quasiprobability distribution, Howard \textit{et al.} \cite{howard2014contextuality} showed that, for discrete-variable systems of odd prime-power dimension, negativity of the (discrete)
Wigner function~\cite{gross2006hudson} corresponds to contextuality with respect
to Pauli measurements. The equivalence was later generalised to odd
dimensions in \cite{delfosse2017equivalence} and to qubit systems in
\cite{raussendorf2017contextuality,delfosse2015wigner}. Under the hypothesis of
noncontextuality, it has also been shown that the
discrete Wigner function is the only possible quasiprobability representation for odd prime dimensions \cite{schmid2021only}.

However, the EPR paradox~\cite{einsteincan1935} and the phase-space description derived by Wigner~\cite{Wigner1932} were originally formulated for CV systems.
Moreover, from a practical point-of-view, CV quantum systems are emerging
as very promising candidates for implementing
quantum informational and computational tasks \cite{braunstein2005quantum,weedbrook2012gaussian,crespi2013integrated,Bourassa2021blueprintscalable,walschaers2021non,chabaud2021continuous} as they offer
unrivalled possibilities for quantum error-correction~\cite{Gottesman2001,cai2021bosonic}, deterministic generation
of large-scale entangled states over millions of subsystems~\cite{yokoyama2013ultra,yoshikawa2016invited}
and reliable and efficient detection methods, such as homodyne or heterodyne detection~\cite{Leonhardt-essential,Ohliger2012}.

Since contextuality and Wigner negativity both seem to play a fundamental role as nonclassical features enabling quantum-over-classical advantages, a natural question arises:

\begin{center}
\textit{What is the precise relationship between contextuality and Wigner negativity in the CV setting?}
\end{center}

Here we prove that contextuality and Wigner negativity are equivalent with respect to CV Pauli measurements, thus unifying the quantum quirks that prevented Einstein and Wigner from obtaining a classically intuitive description of quantum mechanics.
We build on the recent extension of the sheaf-theoretic framework of contextuality~\cite{abramsky2011sheaf} to the CV setting~\cite{barbosa2019continuous}.
Note that this treatment of contextuality is a strict generalisation of the standard notion of Kochen--Specker contextuality \cite{KochenSpecker1967,kochen1975problem}, extended to CV systems.
Using this framework, we prove the equivalence between contextuality and Wigner negativity with respect to generalised position and momentum quadrature measurements, \ie CV Pauli measurements.
These are amongst the most commonly used measurements in CV
quantum information, in particular in quantum optics
\cite{adesso2014continuous,walschaers2021non}, and for defining the standard
models of CV quantum
computing~\cite{lloyd1999quantum,Gottesman2001}.

\section{Phase space and Wigner function}
We fix $M \in \N^*$ to be the number of qumodes, that is, $M$
CV quantum systems. For a single qumode, the corresponding state
space is the Hilbert space of square-integrable functions \(L^2(\R)\) and the total Hilbert space for all \(M\) qumodes is then
\(L^2(\R)^{\otimes M} \cong L^2(\R^{M})\).
To each qumode, we associate the
usual position and momentum operators. We write $\hat q_k$ and $\hat p_k$ the position and momentum operators of the $k^{th}$ qumode. In the context of quantum optics, any linear combination of such operators is called a quadrature of the electromagnetic field~\cite{Leonhardt-essential}. We use this terminology in the rest of the article: any \(\R\)-linear combination of position and momentum operators is called a quadrature. 

The Wigner representation of a quantum state in the Hilbert space \(L^2(\R^{M})\) is a function defined on the phase space \(\R^{2M}\), which can be intuitively understood as a quantum version of the position and momentum phase space of a classical particle. 
We equip this phase space with a symplectic form denoted ${\Omega}$: for $\bm
x,\bm y \in \R^{2M}$, $\Omega(\bm x,\bm y) := \bm x \cdot J \bm y \quad \text{where} \quad J = \left(\begin{smallmatrix}
    0 & \Id_M \\
    - \Id_M & 0 
    \end{smallmatrix}\right)$,
in a given basis \((\bm{e}_k,\bm{f}_k)_{k=1}^{M}\) of \(\R^{2M}\), which is
therefore a symplectic basis for the phase space.
We also equip $\R^{2M}$ with its usual scalar product denoted by $\dummy \cdot \dummy$.

A Lagrangian vector subspace is defined as a maximal isotropic subspace, that is, a maximal subspace on which the symplectic form $\Omega$ vanishes. For a symplectic space of dimension $2M$, Lagrangian subspaces are of dimension $M$. See \cite{Sudarshan1988} for a concise introduction to the symplectic structure of phase space and \cite{Gosson2006symplectic} for a detailed review.

To any \(\bm{x} \in \R^{2M}\) we
associate a quadrature operator as follows. Assume w.l.o.g.\ that \(\bm{x} =
\sum_{k}q_k \bm{e}_k + \sum_{k}p_k \bm{f}_k\), and put $\hat{\bm{x}} = \sum_{k=1}^M q_k \hat{q}_k + \sum_{k=1}^{M} p_k \hat{p}_k$, where the indices indicate on which qumode each operator acts. Then, it is
straightforward to verify, using the canonical commutation relations, that $[\hat{\bm{x}},\hat{\bm{y}}] = i\Omega(\bm{x},\bm{y}) \hat{\Id}$, 
\ie the symplectic structure encodes the commutation relations of quadrature operators.

The elements of \(\R^{2M}\) can also be associated to translations in phase
space. 
Firstly, for any \(s \in \R^{M}\), define the Weyl operators, acting on
\(L^2(\R^{M})\), by $\hat X(s)\psi(t) = \psi(t-s)$ and $\hat Z(s)\psi(t) = e^{ist} \psi(t)$,
for all $t\in\mathbb R^M$.
Then, define the displacement operator for any \(\bm x = (q,p) \in \R^M
\times \R^M\) in the symplectic basis \((\bm{e}_k,\bm{f}_k)_{k=1}^{M}\) by $\hat D(\bm{x}) = e^{-i\frac{q \cdot p}{2}}\hat X(q)\hat Z(p)$,
so that $[\hat D(\bm{x}),\hat D(\bm{y})] = e^{i \Omega(\bm{x},\bm{y})} \hat{\Id}$.

There are several equivalent ways of defining the Wigner function of a quantum state
\cite{ferraro2005gaussian,cahill1969density, de_gosson_wigner_2017}. We follow
the conventions adopted in \cite{de_gosson_symplectic_2011}. The characteristic function \(\Phi_\rho : \R^{2M} \to \C\) of
a density operator \(\rho\) on \(L^2(\R^{M})\) is defined as $\Phi_\rho(\bm{x}) \defeq \Tr(\rho \hat D(-\bm{x}))$.
The Wigner function \(W_\rho\) of \(\rho\) is then defined as the symplectic Fourier transform of the  characteristic function of $\rho$: $W_\rho(\bm{x}) \defeq \operatorname{FT}[\Phi_\rho](J\bm{x})$.
The Wigner function is a real-valued square-integrable function on \(\R^{2M}\), and one can recover the probabilities for
quadrature measurements from its marginals: if \(W\) is the Wigner function of a pure state
\(\psi \in L^2(\R^{M})\) such that \(W\) is integrable on \(\R^{2M}\), then identifying \(\bm{x}\) with \((q,p)\in\R^M\times\R^M\) in the same basis \((\bm
e_k, \bm f_k)_{k=1}^M\) as before,
\begin{align}
    \frac{1}{(\sqrt{2 \pi})^M} \intg{\R^M}{W(q,p)}{p} &= |\psi(q)|^2, \\
    \frac{1}{(\sqrt{2 \pi})^M}  \intg{\R^M}{W(q,p)}{ q} &= \left|\operatorname{FT}[\psi](p)\right|^2.
\end{align}
In general, if \(\bm x \in \R^{2M}\) describes an arbitrary quadrature, the probability of obtaining an outcome $x$ in \(E\subseteq \R\) when measuring the quadrature $\hat{\bm x}$ is
\begin{equation}
  \label{eq:quadrature_probabilities}
  \mathrm{Prob}[x \in E | \rho]
  = \frac{1}{(\sqrt{2\pi})^M} \int_{A} {W_\rho(\bm{y})}\dd{\bm{y}},
\end{equation}
where \(A = \left\{\bm y \in \R^{2M} \mid \bm y \cdot \bm x \in E\right\}\). This corresponds to marginalising the Wigner function over the axes orthogonal to \(\bm{x}\).
If the Wigner function only takes
nonnegative values, it can therefore be interpreted as a simultaneous
probability distribution for position and momentum measurements (and in general, any quadrature obtained as a linear combination of these).

\section{Continuous-variable contextuality}

In what follows, we use the contextuality formalism from \cite{barbosa2019continuous}, which is the extension of \cite{abramsky2011sheaf} to the CV setting. We refer to the Supplemental Material or Refs. \cite{billingsley2008probability,tao2011introduction} for an introduction to this formalism and the associated tools of measure theory. 

\subsection*{Measurement scenario}
In order to define `contextuality' in a CV experiment, we need an abstract description of the experiment, called a measurement scenario, which is defined by a triple $\tuple{\Xc,\Mc,\bm \Oc}$ as follows:
in a given setup, experimenters can choose different measurements to perform on a physical system. Each possible measurement is labelled and $\Xc$ is the corresponding (possibly infinite) set of measurement labels.
Several compatible measurements can be implemented together (for instance, measurements on space-like separated systems). Maximal sets of compatible measurements define a context, and $\Mc$ is the set of all such contexts.
For a measurement labelled by $\bm x\in \Xc$, the corresponding outcome space is $\bm \Oc_{\bm x} = \langle \Oc_{\bm x},\Fc_{\bm x} \rangle$, which is a measurable space with an underlying set $\Oc_{\bm x}$ and its associated Lebesgue $\sigma$-algebra $\Fc_{\bm x}$. The collection of all outcome spaces is denoted $\bm \Oc = (\bm \Oc_{\bm x})_{\bm x \in \Xc}$. For various measurements labelled by elements of a set $U \subseteq X$, the corresponding joint outcome space is denoted $\bm\Oc_U\defeq\prod_{\bm x\in U}\bm \Oc_{\bm x}$. 
In this article, we consider the following measurement scenario: 

\textit{Definition 1.}---The quadrature measurement scenario $\tuple{\Xc,\Mc,\bm \Oc}_{\rm quad}$ is defined as follows: (i) the set of measurement labels is the symplectic phase space $\Xc \defeq \R^{2M}$; (ii) the contexts are Lagrangian subspaces of $\R^{2M}$, so that the set of contexts $\Mc$ is the set of all Lagrangian subspaces of $\Xc$; (iii) for each $\bm x \in \Xc$, the corresponding outcome space is $\bm \Oc_{\bm x} \defeq \langle \R,\sigma
\rangle$ ($\sigma$ being the Lebesgue sigma algebra of $\R$).

$\tuple{\Xc,\Mc,\bm \Oc}_{\rm quad}$ is to be interpreted as follows: given a quantum state $\rho$, the measurement
corresponding to the label $\bm x \in \Xc$ is given by the measurement of the
corresponding quadrature $\hat{\bm{x}}$ of the state, while contexts correspond to maximal sets of commuting quadratures.
This scenario consists in a continuum of possible measurements, each of which corresponds to a quadrature operator with continuous spectrum (see Fig.~\ref{fig:homodyne}).

\subsection*{Empirical model}
While measurement scenarios describe experimental setups, empirical models capture in a precise way the probabilistic behaviours that may arise upon performing measurements on physical systems. In practice, these amount to tables of normalised frequencies of outcomes gathered among various runs of the experiment, or to tables of predicted outcome probabilities obtained by analytical calculation. Formally:

\textit{Definition 2.}---An empirical model on a measurement scenario
$\tuple{\Xc,\Mc,\Oc}$ is a family $e = \family{e_C}_{C \in \Mc}$, where $e_C$ is
a probability measure on the space $\bm \Oc_C$ for each context
$C \in \Mc$.


\subsection*{Noncontextuality and hidden-variable models}
Informally, the empirical data is noncontextual whenever local descriptions (within a valid context) can be glued together consistently so that it can be described by a global probability measure (over all contexts).

\textit{Definition 3.}---An empirical model $e=\family{e_C}_{C \in \Mc}$ on a $\tuple{\Xc,\Mc,\bm \Oc}$ is noncontextual if there exists a probability measure $p$ on the space $\bm \Oc_\Xc$ such that marginalising $p$ on a context gives back the empirical prediction \ie $p|_C = e_C$ for every $C \in \Mc$.

Noncontextuality is equivalent to the existence of a deterministic hidden-variable model (HVM) \cite{barbosa2019continuous}.

\textit{Definition 4.}---A HVM on a measurement scenario $\tuple{\Xc,\Mc,\bm\Oc}$ is a tuple
  $\tuple{\bm \Lambda, p, (k_C)_{C \in \Mc}}$ where: (i)
   $\bfLambda = \tuple{\Lambda,\Fc_\Lambda}$ is the measurable space of hidden variables; 
  (ii) $p$ is a probability distribution on $\bfLambda$;
  (iii) for each context $C \in \Mc$, $k_C$ is a probability kernel between the measurable spaces $\bfLambda$ and $\bm \Oc_C$, i.e.\ $k_C$ is a measurable function over $\bfLambda$ and a probability measure over $\bm \Oc_C$.

Determinism for the HVM is further ensured by requiring that each hidden variable gives a predetermined outcome. That is, for all contexts $C \in \Mc$ and for every $\lambda \in \Lambda$, $k_C(\lambda,\dummy) = \delta_{\bm x}$ is a Dirac measure at some $\bm x \in \Oc_C$.

The space $\bm \Oc_\Xc$ can be thought of as a space of hidden variables, while the probability measure $p$ provides probabilistic information about them. 
Hidden variables are supposed to provide an underlying description of the physical world at perhaps a more fundamental level than the empirical-level description via the quantum state. The motivation is that hidden variables could explain away some of the non-intuitive aspects of the empirical predictions of quantum mechanics, which would then arise from an incomplete knowledge of the true state of a system rather than being fundamental features. 

\subsection*{Quantum models and linear assignments}
Since we consider experiments arising from quadrature measurements of a quantum system $\rho$, we restrict our attention to empirical models $e = (e_C)_{C \in \Mc}$ reproducing the Born rule, which we refer to as quantum empirical models. We will use the notation $e^{\rho} = (e_C^\rho)_{C \in \Mc}$ to make explicit the dependence in the state $\rho$. If $e^\rho$ is noncontextual in $\tuple{\Xc,\Mc,\bm \Oc}_{\rm quad}$, we say that the density operator $\rho$ is noncontextual for quadrature measurements.

In $\tuple{\Xc,\Mc,\bm \Oc}_{\rm quad}$, for each context $C \in \Mc$, the set $\Oc_C = \prod_{\bm x \in C} \R$ can be seen as the set of functions from $C$ to $\R$ with the corresponding product $\sigma$-algebra. These functions are called  local value assignments. In contrast, functions \(\Xc \to \R\) which assign a tentative outcome to all quadratures simultaneously are called global value assignments. Contextuality then expresses the tension which arises when trying to explain different experimental predictions across distinct contexts (local value assignments) in terms of global value assignments.

\textit{Linearity of value assignments.}---
Before connecting the notions of contextuality and negativity of the Wigner function in $\tuple{\Xc,\Mc,\bm \Oc}_{\rm quad}$, we must resolve the following mismatch: the Wigner function of a density operator $\rho$ is a quasiprobability
distribution over $\Xc=\R^{2M}$, while a density operator $\rho$ which is noncontextual for quadrature measurements corresponds by Definition 3 to
a global probability measure over the set $\Oc_\Xc$ of global assignments, which can be seen in this case as the set of
functions $\Xc=\R^{2M} \to \R$. In general, the latter is much larger than the former.

To solve this issue, we show that we can restrict to linear global value assignments w.l.o.g. so that $\Oc_\Xc$ can be replaced by $\Xc^*$, the linear dual of $\Xc$. Since $\Xc^*$ is isomorphic to $\Xc$, this allows us to resolve
the mismatch:

\noindent\textit{Proposition 1.}---If $M \geqslant 2$, global value assignments can w.l.o.g. be taken to be linear functions $\Xc \to \R$, and the set of global value assignments then forms an $\R$-linear space of dimension $2M$, namely $\Xc^*$.

We refer to Appendix A for the proof. Therefore w.l.o.g. for any $U \subset \Xc$, we can restrict $\Oc_U$ to be the set of functions from $U \to \R$ that admit a (not necessarily unique) extension by an $\R$-linear function on $\Xc$.

\section{Equivalence between quantum contextuality and Wigner negativity}
We may now prove our main result, \ie the equivalence between contextuality and Wigner negativity in the quadrature measurement scenario $\tuple{\Xc,\Mc,\bm \Oc}_{\rm quad}$:

\noindent\textit{Theorem 1.}---For any \(M \geqslant 2\), a density operator \(\rho\) on \(L^2(\R^M)\) is noncontextual for quadrature measurements if and only if its Wigner function $W_\rho$ is both integrable and nonnegative.

We refer to Appendix B for the full proof.

\textit{Sketch of proof.}---Using Definition 3 and Proposition 1, a density operator $\rho$ which is noncontextual for quadrature measurements corresponds to a probability measure $p$ on the space $\Xc^*$ describing the outcomes of quadrature measurements on $\rho$. We show that the Fourier transform of $p$ must be equal to the characteristic function $\Phi_\rho$. Hence, $p$ and  $W_\rho$ have the same Fourier transform, which gives $w_p=W_\rho$ and thus $W_\rho\ge0$, since $w_p$ is the density of a probability measure. Conversely, if the Wigner function is assumed to be integrable and nonnegative, we obtain the outcome probabilities for quadrature measurements by marginalising along the correct axes.\qed

\textit{Experimental setup.}---The quadrature measurement scenario requires measuring any linear combination of multimode position and momentum operators, \eg $\hat q_1 + 2 \hat p_1 + 5 \hat q_2$. The measurement setup is represented in Fig.~\ref{fig:homodyne}.

To do so experimentally, we first apply phase-shift operators $\hat R$ for each individual qumode to obtain the right quadrature for each mode.
Then, we apply CZ gates of the form $e^{i g \hat q_k \hat q_l}$ for $g \in \R$ to pairs of qumodes $k$ and $l$ in order to sum them. 
This permits the construction of the desired linear combination, which is stored in one quadrature of a qumode. The measurement can then be implemented with standard homodyne detection,
which consists in a Gaussian measurement of a quadrature of the field, by mixing the state on a balanced beam splitter (dashed line) with a strong coherent state (local osccilator LO).
The intensities of both output arms are measured with photodiode detectors and their difference yields a value proportional to a quadrature $\hat x_\phi:=(\cos\phi)\hat q+(\sin\phi)\hat p$ of the input qumode, depending on the phase $\phi$ of the local oscillator.
All of these steps can be implemented experimentally with current optical technology \cite{ferraro2005gaussian,su2013gate}. 

\begin{figure}
    \centering
    \includegraphics[scale=0.9]{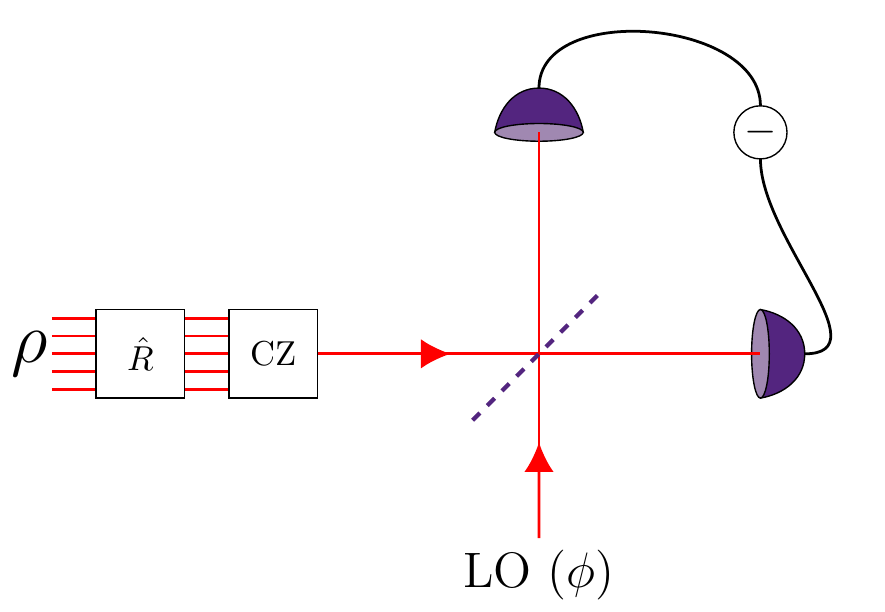}
    \caption{Experimental protocol corresponding to the quadrature measurement scenario $\tuple{\Xc,\Mc,\bm \Oc}_{\rm quad}$. 
    }
    \label{fig:homodyne}
\end{figure}

\section{Discussion}

We have shown that Wigner negativity is equivalent to contextuality with respect to CV quantum measurements which may be realised using homodyne detection, a standard detection method in CV \cite{yokoyama2013ultra}, and the basis of several computational models in CV quantum information \cite{DouceCVIQP2017,Chabaud2017hom,BShomodyne2017,GKP2001}.

From a practical perspective, this implies that contextuality is a necessary resource for achieving a computational advantage within the standard model of CV quantum computation~\cite{lloyd1999quantum}. Like in the discrete-variable case \cite{howard2014contextuality}, CV contextuality supplies the necessary ingredients for CV quantum computing.

From a foundational perspective, the failure of a local hidden-variable model describing quantum mechanical predictions, as enlightened by Bell regarding the EPR paradox, is very closely related to the impossibility of a nonnegative phase-space distribution, as described by Wigner. 
Hence, our result implies that the negativity of phase-space distributions can be cast as an obstruction to the existence of a noncontextual hidden-variable model.

The EPR state \cite{einsteincan1935} describes a CV state that has a nonnegative Wigner function and still violates a Bell inequality \cite{Banaszek1998}. This is possible since it necessitates parity operator measurements that do not have a nonnegative Wigner representation \cite{spekkens2008negativity}, and thus is not in contradiction with our result. Indeed, our quadrature measurement scenario is nonnegatively represented in phase space: since homodyne detection (and quadrature measurements in general) is a Gaussian measurement, any possible quantum advantage is due to Wigner negativity being present before the detection setup. 

Our results open up a number of future research directions. Firstly, our present argument requires considering a measurement scenario that comprises an uncountable family of measurement labels (the entire phase space $\Xc=\R^{2M}$). From an experimental perspective, it is crucial to wonder what happens if we restrict to a finite family of measurement labels and see whether we can derive a robust version of this theorem. 

Another question concerns the link between quantifying contextuality and quantifying Wigner negativity. Quantifying contextuality for CV systems is possible via semidefinite relaxation \cite{barbosa2015contextuality}. Also, there exist various measures of Wigner negativity~\cite{kenfack2004negativity,mari2011directly}.
In particular, witnesses for Wigner negativity have been introduced in \cite{chabaud2021witnessing}, whose violation gives a lower bound on the distance to the set of states with nonnegative Wigner function. It would be highly desirable to establish a precise and quantified link between these different measures of nonclassicality.

This equivalence may also be useful
in better understanding the problem of characterising those quantum-mechanical
states whose Wigner function is nonnegative. This is a notoriously thorny issue
when one considers mixed states, and, to the authors' knowledge, progress has
mostly stalled since the 90s \cite{narcowich_necessary_1986,
  narcowich_unified_1988, brocker_mixed_1995, de_gosson_symplectic_2011}. Strong
mathematical tools are being developed to detect contextuality from the
theoretical description of a state \cite{abramsky2012cohomology,
  abramsky2015contextuality,abramsky2017contextual, caru2017, caru2018towards,
  raussendorf_cohomological_2019, okay_classifying_2021}, although much work
remains in applying them to CV and understanding their
relation to machinery previously developed to tackle the positivity question.

On the practical side, our result paves the way for surprisingly simple
demonstrations of nonclassicality. Contextual states are typically associated
with violated Bell-like inequalities---although this result has only been
formally proven in the case of a finite number of measurement settings
\cite{barbosa2019continuous}, and needs to be generalised to CV. In principle, this means that one should be able to violate such an
inequality with a setup as simple as a single photon and a heterodyne detection,
necessitating only a single beamsplitter. The existence of such a genuinely
continuous Bell inequality has been elusive, since previously-observed
violations amount to encodings of discrete-variable inequalities in
CV~\cite{plastino2010state, he2010bell, mckeown2011testing,
  asadian2015contextuality, laversanne2017general, ketterer2018continuous,
  Fred2003}.


\paragraph*{Acknowledgements.}
The authors would like to thank M.\ Howard, S.\ Mansfield and  T.\ Douce for enlightening discussions. They are also grateful to D.\ Markham, E.\ Kashefi, E.\ Diamanti and F.\ Grosshans for their mentorship and the wonderful group they have created at LIP6.
UC acknowledges interesting discussions with S.\ Mehraban and J.\ Preskill.
UC acknowledges funding provided by the Institute for Quantum Information
and Matter, a National Science Foundation Physics Frontiers Center (NSF Grant PHY-1733907).
RIB was supported by the Agence Nationale de la Recherche VanQuTe project (ANR-17-CE24-0035).

\paragraph*{Related work.} An analogous, independent proof of our main theorem was subsequently derived in \cite{haferkamp2021equivalence}.

\medskip

\section*{Proof Proposition 1}

\noindent\textit{Proposition 1.}---If $M \geqslant 2$, global value assignments can w.l.o.g. be taken to be linear functions $\Xc \to \R$, and the set of global value assignments then forms an $\R$-linear space of dimension $2M$, namely $\Xc^*$.

\noindent \textit{Proof of Proposition~1.}---Let $L \subseteq \Xc$ be a Lagrangian subspace and let $e^\rho$ be a quantum empirical model. For $U \in \Fc_L$ a Lebesgue measurable set of functions $L \to \R$ and for $\bm x \in L$, let $\pi_{\bm x}(U):= \left\{f(\bm x) \mid f \in U \right\} \subseteq \R$. We start by showing the following result:

\noindent\textit{Lemma 1.}---
Let $U \in \Fc_L$ be a Lebesgue measurable set of functions $L \to \R$ such that
$\pi_{\bm x}(U)$ is distinct from $\R$ for a finite number of $\bm x \in L$.
Then there exists a subset $U_{\mathrm{lin}}$ of linear functions $L \to \R$
such that for all $\bm x \in L$, $\pi_{\bm x}(U_{\mathrm{lin}}) \subseteq
\pi_{\bm x}(U)$ and \(e_L^\rho(U_{\mathrm{lin}}) = e_L^\rho(U)\).

\textit{Proof.}---First, let $(\bm e_k)_{k=1,\dots,M}$ be a basis of $L \cong \R^M$. Let $P$ be the joint spectral measure of $\{{\bm{ \hat e_1}},\dots,{\bm{ \hat e_M}}\}$.
For any $ \bm y \in L$, define the function
    \begin{align}
        f_{\bm y} : L & \longrightarrow \R \\
              \bm x & \longmapsto \bm x \cdot \bm y \, ,
    \end{align}
where $\dummy \cdot \dummy $ is the usual Euclidean scalar product on $L \cong \R^M$.
For any $\bm x \in L$, $P_{\bm {\hat x}}$ is the push-forward of $P$ by the measurable function $f_{\bm x}$ by the functional calculus on $M$ commuting observables (see Supplemental Material for a concise introduction to measure theory; this includes Ref. \cite{tychonoff1930topologische} for the product $\sigma$-algebra).

Then,
\begin{align}
    \Tr\!\left(\!\rho \prod_{\bm x \in L}\!P_{\bm {\hat x}} \circ \pi_{\bm x}(U)\!\right) 
    \!&=\!\Tr\!\left(\!\rho \prod_{\bm x \in L}\!P \left(f_{\bm x}^{-1}\left( \pi_{\bm x}(U)\right)\right)\!\right) \\ 
    \!&=\!\Tr\!\left(\!\rho  P\!\left(\bigcap_{\bm x \in L} f_{\bm x}^{-1}\!\left(\pi_{\bm x}(U)\right)\!\right)\!\right)\!
\end{align}
with
\begin{align}
    \bigcap_{\bm x \in L} f_{\bm x}^{-1}(\pi_{\bm x}(U)) = \left\{ \bm y \in L \mid \forall \bm x \in L, \, \bm x \cdot \bm y \in \pi_{\bm x}(U) \right\}.
\end{align}
Now define
\begin{equation}
    U_{\mathrm{lin}} \defeq \left\{
      \begin{aligned}
        L &\longrightarrow \R \\
        \bm x &\longmapsto \bm x \cdot \bm y
      \end{aligned}
      \;\bigg\vert\;\bm y \in \bigcap_{\bm x \in L} f_{\bm x}^{-1}(\pi_{\bm x}(U)) 
    \right\}\!.
\end{equation}
By construction,
\begin{widetext}
\begin{align}
    \bigcap_{\bm x \in L} f_{\bm x}^{-1}(\pi_{\bm x}(U_{\mathrm{lin}}))
    &= \bigcap_{\bm x \in L} f_{\bm x}^{-1} \left( \left\{ \bm x \cdot \bm y \mid \bm y \in L \text{ s.t. } \forall \bm z \in L, \, \bm y \cdot \bm z \in \pi_{\bm z}(U) \right\} \right)\\
    &= \bigcap_{\bm x \in L}  \left\{ \bm \alpha \in L \mid \bm x \cdot \bm \alpha = \bm x \cdot \bm y \text{ with } \bm y \in L \text{ s.t. } \forall \bm z \in L, \, \bm y \cdot \bm z \in \pi_{\bm z}(U) \right\} \\
    &= \left\{ \bm \alpha \in L \mid \forall \bm x \in L, \, \bm x \cdot \bm \alpha = \bm x \cdot \bm y \text{ with } \bm y \in L \text{ s.t. } \forall \bm z \in L, \, \bm y \cdot \bm z \in \pi_{\bm z}(U) \right\} \\
    &= \left\{ \bm \alpha \in L \mid \forall \bm z \in L, \, \bm \alpha \cdot \bm z \in \pi_{\bm z}(U) \right\} \label{eq:equality}\\
    &= \bigcap_{\bm x \in L} f_{\bm x}^{-1}(\pi_{\bm x}(U))\, , 
\end{align}
\end{widetext}
where \eqref{eq:equality} follows from that fact that $(\forall \bm x \in L$, $\bm x \cdot  \bm \alpha = \bm x \cdot \bm y)$ implies $\bm \alpha = \bm y$.
Also for all $\bm x \in L$, $\pi_{\bm x}(U_{\mathrm{lin}}) \subseteq \pi_{\bm x}(U)$ so that we are indeed reproducing all value assignments from linear functions of $U$. Then, by the Born rule,
  \begin{align}
    e_L^\rho (U_{\mathrm{lin}})
    &= \Tr\left(\rho \prod_{\bm x \in L} P_{\bm {\hat x}} \circ \pi_{\bm x}(U_{\mathrm{lin}})\right) \\
    &= \Tr\left(\rho \prod_{\bm x \in L} P \left( f_{\bm x}^{-1} \left( \pi_{\bm x}(U_{\mathrm{lin}})\right)\right) \right)\\
    &= \Tr\left(\rho  P\left(\bigcap_{\bm x \in L} f_{\bm x}^{-1}\left( \pi_{\bm x}(U_{\mathrm{lin}} \right)\right)\right) \\
    &= \Tr\left(\rho  P\left(\bigcap_{\bm x \in L} f_{\bm x}^{-1}\left(\pi_{\bm x}(U)\right)\right)\right) \\
    &= \Tr\left(\rho \prod_{\bm x \in L} P_{\bm {\hat x}}\left(\pi_{\bm x}(U)\right)\right)\\
    &= e_L^\rho (U).
  \end{align}
\hfill\qed

\noindent We are now in position to prove Proposition~1.

The sheaf-theoretic framework for contextuality describes value assignments as
  a sheaf $\mathscr{E}$ where
  $\mathscr{E}(U)$ is the set of value assignments for the measurement labels
  in $U$, which can be viewed as a set of functions $U \to \R$. For any Lagrangian $L
  \in \Mc$, there is a restriction map $\mathscr{E}(\Xc) \to
  \mathscr{E}(L) : f \mapsto f|_L$ that simply restricts the domain of any function from $\Xc$ to $L$. Then $\mathscr{E}(L)$ must coincide with
  the set of possible value assignments $\Oc_L$.

  By Lemma~1, $\mathscr{E}(L)$ consists in linear
  functions $L \to \R$ so that the set of global value assignments $\mathscr{E}(\Xc)$ contains only functions $\Xc
  \to \R$ whose restriction to any Lagrangian subspace is $\R$-linear. Then,
  following \cite[Lemma~1]{delfosse2017equivalence} (the Lemma is proven for the
  discrete phase-space $\Z_d^M \times \Z_d^M$ but its proof extends directly to
  $\R^M \times \R^M$), we conclude that if $M \geqslant 2$, $\mathscr{E}(\Xc)$
  contains only $\R$-linear functions $\Xc \to \R$, i.e. $\mathscr{E}(\Xc) =
  \Xc^*$, where $\Xc^*$ is the dual space of $\Xc$.
\qed

\medskip

\section*{Proof Theorem 1}
\noindent\textit{Theorem 1.}---For any \(M \geqslant 2\), a density operator \(\rho\) on \(L^2(\R^M)\) is noncontextual for quadrature measurements if and only if its Wigner function $W_\rho$ is both integrable and nonnegative.

\noindent \textit{Proof of Theorem~1.}---The proof proceeds by showing both directions of the equivalence. We make use of standard properties of the Wigner function, which are reviewed in the Supplemental Material.

($\implies$) 
It follows from the definition of the Wigner function that, for $\bm x \in \Xc$, $\Phi_\rho(\bm x) = \text{FT}^{-1}
\left[W_\rho \right](-J\bm x) = \text{FT} \left[W_\rho \right](J\bm x)$.
On the other hand, fix a noncontextual quantum empirical model $e^\rho$ satisfying the
Born rule associated to $\rho$ and the quadrature measurement scenario $\tuple{\Xc,\Mc,\bm \Oc}_{\rm quad}$. 
By Proposition~1, the set of global value assignments forms an $\R$-linear space of dimension $2M$, namely $\Xc^*$, which is isomorphic to its dual $\Xc$. Hence, by \cite[Theorem~1]{barbosa2019continuous} (restated in the Supplemental Material as Proposition~13) we
have a deterministic HVM \(\tuple{\bm\Xc, p, (k_C)_{C \in \Mc}}\) for \(e^\rho\) (see Definition~4), where $\bm\Xc= \tuple{\Xc,\Fc_\Xc}$ with $\Fc_\Xc$ the usual $\sigma$-algebra of $\Xc=\R^{2M}$. Moreover:
  \begin{align}
  \Phi_\rho(\bm x) &= \Tr\left(\hat D(-\bm x)\rho\right) \\
  &= \Tr\left(\rho \intg{\lambda\in\R}{e^{-i\lambda}}{P_{\hat{J\bm x}}(\lambda)}\right) \\
  &= \intg{\R} {e^{-i\lambda}} {p_{\hat{J\bm{x}}}^\rho(\lambda)} \\
  &= \intg{\Xc} {e^{-i J\bm x \cdot \bm y}}{p (\bm y)} \\
  &= \text{FT}[p](J\bm x),
    \end{align}
where the second line comes from the spectral theorem; the third line by
letting \(p_{\hat{J \bm x}}^\rho(E) = \Tr(P_{\hat{J\bm x}}(E)\rho)\) and
the fact that the integral and the trace may be inverted by the definition of
the integral with respect to the spectral measure \cite{hall2013quantum}; the
fourth line via the push-forward of measures; and the last line comes the
definition of the Fourier transform of a measure. 

Since the characteristic function \(\Phi_\rho\) is square-integrable, we can apply \cite[Lemma~1.1]{fournier_absolute_2010} (restated in the Supplemental Material as Lemma~7). Then, the measure \(p\) must have a density \(w_p
\in L^2(\R^M)\).
As a result, for all \(x \in \Xc\), $\text{FT}[w_p](\bm x) =
\text{FT}[p](\bm x) = \Phi_\rho(-J \bm x) = \text{FT}[W_\rho](\bm x)$
and since $w_p$ and $W_\rho$ are both in $L^2(\Xc)$ on which the Fourier
transform is unitary, it must hold that $w_p = W_\rho$ $\dd{\bm{x}}$-almost
everywhere. $w_p$ is the density of a probability measure, so it follows that
both functions must be almost everywhere nonnegative. Because the Wigner
function is a continuous function from $\Xc$ to $\R$ \cite{cahill1969density},
$W_\rho$ must be nonnegative.

($\impliedby$) Conversely, the Wigner function provides the correct marginals
for the quadratures and can be seen as a global probability density on phase
space when it is nonnegative. Via the equivalence demonstrated in the first
part of the proof (namely that the density of $p$ is almost everywhere the
Wigner function), the idea is to show that $\tuple{\bm \Xc,W_{\rho} \dd{\bm
  x},(\smash{\tilde k_L})_{L\in\Mc}}$ is a valid deterministic HVM (see Definition 4) that reproduces the
empirical predictions, where \(\tilde{k}_L(\bm x,U) \coloneqq (\sqrt{2\pi})^{-M} \delta_{(\bm
  x \cdot \dummy)|_L}(U)\).

For any \(\bm x \in \Xc\), there is a special orthogonal and symplectic
transformation \(S\) such that \(\bm x = \|\bm x\| S\bm e_1\). For any \(U \in
\Fc_{\bm{x}}\),
     \begin{align}
      e_{\bm{x}}^\rho(U) &= \Tr(P_{\hat{\bm{x}}}\circ\pi_{\bm{x}}(U)\rho) \\
      & = \frac{1}{(\sqrt{2\pi})^M} \intg{\|\bm x\|^{-1} \pi_{\bm{x}}(U) \times \R^{2M-1}}{W_\rho(S \bm{z})}{\bm{z}} \\
      & = \frac{1}{(\sqrt{2\pi})^M} \intg{(\bm e_1 \cdot -)^{-1}(\|\bm x\|^{-1} \cdot \pi_{\bm{x}}(U))}{\hspace{-2mm}W_\rho(S \bm{z})}{\bm{z}} \\
      & = \frac{1}{(\sqrt{2\pi})^M} \intg{(\|\bm x\| \bm e_1 \cdot S^{-1}-)^{-1}(\pi_{\bm{x}}(U))}{W_\rho(\bm{z})}{\bm{z}}  \label{eq:changeofvariablesJac}\\
      & = \frac{1}{(\sqrt{2\pi})^M} \intg{(\|\bm x\| S\bm e_1 \cdot -)^{-1}(\pi_{\bm{x}}(U))}{W_\rho(\bm{z})}{\bm{z}} \\
      & = \frac{1}{(\sqrt{2\pi})^M} \intg{(\bm x \cdot -)^{-1}(\pi_{\bm{x}}(U))}{W_\rho(\bm{z})}{\bm{z}},
   \end{align}
  where we have used the symplectic covariance of the Wigner function (restated as Lemma~9 in the Supplemental Material) and the fact that the Jacobian change of
  variable in \eqref{eq:changeofvariablesJac} is 1.
By definition of $\tilde k$, for all $\bm x\in\Xc$ and all $U\in
\Fc_{\bm{x}}$:
    \begin{align}
     &\intg{\bm z \in \Xc}{\tilde k_{\{\bm x\}}(\bm z,U) W_{\rho}(\bm z)}{\bm z} \\
     &\quad\quad= \frac{1}{(\sqrt{2\pi})^M} \intg{\bm z \in \Xc}{\delta_{(\bm z \cdot \dummy)|_{\{ \bm x \}}}(U) W_{\rho}(\bm z)}{\bm z}.
    \end{align}
$U\in\Fc_{\bm{x}}$ consists of functions \(\{\bm x\} \to \R\), which therefore extend to linear functions on the subspace generated by $\bm x$, so:
    \begin{align}
 \{\bm z \in \Xc \mid (\bm z \cdot \dummy)|_{\{\bm x\}} \in U\} 
 &= \{\bm z \in \Xc \mid (\bm z\cdot \bm x) \in \pi_{\bm x}(U)  \} \\
 &= (\dummy \cdot \bm x)^{-1}(\pi_{\bm x}(U)) \label{eq:cov}.
\end{align}
Thus:
    \begin{align}
    \intg{\bm z \in \Xc}{\hspace{-2mm}\tilde k_{\{\bm x\}}(\bm z,U) W_{\rho}(\bm z)}{ \bm z} 
    &= \frac{1}{(\sqrt{2\pi})^M} \intg{(\dummy \cdot \bm x)^{-1}(\pi_{\bm x}(U))}{ \hspace{-10mm}W_{\rho}(\bm z) }{\bm z}\\
    &= e_{\bm{x}}^\rho(U),
\end{align}
 While we have verified the calculation only for \(e_{\bm x}^\rho\)(U) for $U\in \Fc_{\bm x}$, the same computation can be carried out for $e^\rho_L(U)$ for a Lagrangian subspace $L$ and $U \in \Fc_L$ to retrieve the joint probability distributions from the HVM $\tuple{\bm \Xc,W_{\rho} \dd{\bm x},(\smash{\tilde k_L})_{L\in\Mc}}$, using Proposition~1 to restrict the elements of $U$ to linear functions.
 \qed


\begin{thebibliography}{80}%
\makeatletter
\providecommand \@ifxundefined [1]{%
 \@ifx{#1\undefined}
}%
\providecommand \@ifnum [1]{%
 \ifnum #1\expandafter \@firstoftwo
 \else \expandafter \@secondoftwo
 \fi
}%
\providecommand \@ifx [1]{%
 \ifx #1\expandafter \@firstoftwo
 \else \expandafter \@secondoftwo
 \fi
}%
\providecommand \natexlab [1]{#1}%
\providecommand \enquote  [1]{``#1''}%
\providecommand \bibnamefont  [1]{#1}%
\providecommand \bibfnamefont [1]{#1}%
\providecommand \citenamefont [1]{#1}%
\providecommand \href@noop [0]{\@secondoftwo}%
\providecommand \href [0]{\begingroup \@sanitize@url \@href}%
\providecommand \@href[1]{\@@startlink{#1}\@@href}%
\providecommand \@@href[1]{\endgroup#1\@@endlink}%
\providecommand \@sanitize@url [0]{\catcode `\\12\catcode `\$12\catcode
  `\&12\catcode `\#12\catcode `\^12\catcode `\_12\catcode `\%12\relax}%
\providecommand \@@startlink[1]{}%
\providecommand \@@endlink[0]{}%
\providecommand \url  [0]{\begingroup\@sanitize@url \@url }%
\providecommand \@url [1]{\endgroup\@href {#1}{\urlprefix }}%
\providecommand \urlprefix  [0]{URL }%
\providecommand \Eprint [0]{\href }%
\providecommand \doibase [0]{https://doi.org/}%
\providecommand \selectlanguage [0]{\@gobble}%
\providecommand \bibinfo  [0]{\@secondoftwo}%
\providecommand \bibfield  [0]{\@secondoftwo}%
\providecommand \translation [1]{[#1]}%
\providecommand \BibitemOpen [0]{}%
\providecommand \bibitemStop [0]{}%
\providecommand \bibitemNoStop [0]{.\EOS\space}%
\providecommand \EOS [0]{\spacefactor3000\relax}%
\providecommand \BibitemShut  [1]{\csname bibitem#1\endcsname}%
\let\auto@bib@innerbib\@empty
\bibitem [{\citenamefont {Arute}\ \emph {et~al.}(2019)\citenamefont {Arute},
  \citenamefont {Arya}, \citenamefont {Babbush}, \citenamefont {Bacon},
  \citenamefont {Bardin}, \citenamefont {Barends}, \citenamefont {Biswas},
  \citenamefont {Boixo}, \citenamefont {Brandao}, \citenamefont {Buell} \emph
  {et~al.}}]{arute2019quantum}%
  \BibitemOpen
  \bibfield  {author} {\bibinfo {author} {\bibfnamefont {F.}~\bibnamefont
  {Arute}}, \bibinfo {author} {\bibfnamefont {K.}~\bibnamefont {Arya}},
  \bibinfo {author} {\bibfnamefont {R.}~\bibnamefont {Babbush}}, \bibinfo
  {author} {\bibfnamefont {D.}~\bibnamefont {Bacon}}, \bibinfo {author}
  {\bibfnamefont {J.~C.}\ \bibnamefont {Bardin}}, \bibinfo {author}
  {\bibfnamefont {R.}~\bibnamefont {Barends}}, \bibinfo {author} {\bibfnamefont
  {R.}~\bibnamefont {Biswas}}, \bibinfo {author} {\bibfnamefont
  {S.}~\bibnamefont {Boixo}}, \bibinfo {author} {\bibfnamefont {F.~G.}\
  \bibnamefont {Brandao}}, \bibinfo {author} {\bibfnamefont {D.~A.}\
  \bibnamefont {Buell}}, \emph {et~al.},\ }\bibfield  {title} {\bibinfo {title}
  {Quantum supremacy using a programmable superconducting processor},\ }\href
  {https://doi.org/10.1117/12.2603523} {\bibfield  {journal} {\bibinfo
  {journal} {Nature}\ }\textbf {\bibinfo {volume} {574}},\ \bibinfo {pages}
  {505} (\bibinfo {year} {2019})}\BibitemShut {NoStop}%
\bibitem [{\citenamefont {Zhong}\ \emph {et~al.}(2020)\citenamefont {Zhong},
  \citenamefont {Wang}, \citenamefont {Deng}, \citenamefont {Chen},
  \citenamefont {Peng}, \citenamefont {Luo}, \citenamefont {Qin}, \citenamefont
  {Wu}, \citenamefont {Ding}, \citenamefont {Hu} \emph
  {et~al.}}]{zhong2020quantum}%
  \BibitemOpen
  \bibfield  {author} {\bibinfo {author} {\bibfnamefont {H.-S.}\ \bibnamefont
  {Zhong}}, \bibinfo {author} {\bibfnamefont {H.}~\bibnamefont {Wang}},
  \bibinfo {author} {\bibfnamefont {Y.-H.}\ \bibnamefont {Deng}}, \bibinfo
  {author} {\bibfnamefont {M.-C.}\ \bibnamefont {Chen}}, \bibinfo {author}
  {\bibfnamefont {L.-C.}\ \bibnamefont {Peng}}, \bibinfo {author}
  {\bibfnamefont {Y.-H.}\ \bibnamefont {Luo}}, \bibinfo {author} {\bibfnamefont
  {J.}~\bibnamefont {Qin}}, \bibinfo {author} {\bibfnamefont {D.}~\bibnamefont
  {Wu}}, \bibinfo {author} {\bibfnamefont {X.}~\bibnamefont {Ding}}, \bibinfo
  {author} {\bibfnamefont {Y.}~\bibnamefont {Hu}}, \emph {et~al.},\ }\bibfield
  {title} {\bibinfo {title} {Quantum computational advantage using photons},\
  }\href {https://doi.org/10.1017/cbo9780511622748.004} {\bibfield  {journal}
  {\bibinfo  {journal} {Science}\ }\textbf {\bibinfo {volume} {370}},\ \bibinfo
  {pages} {1460} (\bibinfo {year} {2020})}\BibitemShut {NoStop}%
\bibitem [{\citenamefont {Einstein}\ \emph {et~al.}(1935)\citenamefont
  {Einstein}, \citenamefont {Podolsky},\ and\ \citenamefont
  {Rosen}}]{einsteincan1935}%
  \BibitemOpen
  \bibfield  {author} {\bibinfo {author} {\bibfnamefont {A.}~\bibnamefont
  {Einstein}}, \bibinfo {author} {\bibfnamefont {B.}~\bibnamefont {Podolsky}},\
  and\ \bibinfo {author} {\bibfnamefont {N.}~\bibnamefont {Rosen}},\ }\bibfield
   {title} {\bibinfo {title} {Can {quantum}-{mechanical} {description} of
  {physical} {reality} {be} {considered} {complete}?},\ }\href
  {https://doi.org/10.1103/PhysRev.47.777} {\bibfield  {journal} {\bibinfo
  {journal} {Physical Review}\ }\textbf {\bibinfo {volume} {47}},\ \bibinfo
  {pages} {777} (\bibinfo {year} {1935})}\BibitemShut {NoStop}%
\bibitem [{\citenamefont {{W}igner}(1932)}]{Wigner1932}%
  \BibitemOpen
  \bibfield  {author} {\bibinfo {author} {\bibfnamefont {E.}~\bibnamefont
  {{W}igner}},\ }\bibfield  {title} {\bibinfo {title} {On the quantum
  correction for thermodynamic equilibrium},\ }\href
  {https://doi.org/10.1103/PhysRev.40.749} {\bibfield  {journal} {\bibinfo
  {journal} {Physical Review}\ }\textbf {\bibinfo {volume} {40}},\ \bibinfo
  {pages} {749} (\bibinfo {year} {1932})}\BibitemShut {NoStop}%
\bibitem [{\citenamefont {Bell}(1964)}]{belleinstein1964}%
  \BibitemOpen
  \bibfield  {author} {\bibinfo {author} {\bibfnamefont {J.~S.}\ \bibnamefont
  {Bell}},\ }\bibfield  {title} {\bibinfo {title} {On the {Einstein} {Podolsky}
  {Rosen} paradox},\ }\href
  {https://doi.org/10.1103/PhysicsPhysiqueFizika.1.195} {\bibfield  {journal}
  {\bibinfo  {journal} {Physics Physique Fizika}\ }\textbf {\bibinfo {volume}
  {1}},\ \bibinfo {pages} {195} (\bibinfo {year} {1964})},\ \bibinfo {note}
  {publisher: American Physical Society}\BibitemShut {NoStop}%
\bibitem [{\citenamefont {Bell}(1966)}]{bell1966}%
  \BibitemOpen
  \bibfield  {author} {\bibinfo {author} {\bibfnamefont {J.~S.}\ \bibnamefont
  {Bell}},\ }\bibfield  {title} {\bibinfo {title} {On the problem of hidden
  variables in quantum mechanics},\ }\href
  {https://doi.org/10.1103/RevModPhys.38.447} {\bibfield  {journal} {\bibinfo
  {journal} {Reviews of Modern Physics}\ }\textbf {\bibinfo {volume} {38}},\
  \bibinfo {pages} {447} (\bibinfo {year} {1966})}\BibitemShut {NoStop}%
\bibitem [{\citenamefont {Kochen}\ and\ \citenamefont
  {Specker}(1975)}]{kochen1975problem}%
  \BibitemOpen
  \bibfield  {author} {\bibinfo {author} {\bibfnamefont {S.}~\bibnamefont
  {Kochen}}\ and\ \bibinfo {author} {\bibfnamefont {E.~P.}\ \bibnamefont
  {Specker}},\ }\bibfield  {title} {\bibinfo {title} {The problem of hidden
  variables in quantum mechanics},\ }in\ \href
  {https://doi.org/10.1007/978-94-010-1795-4_17} {\emph {\bibinfo {booktitle}
  {The logico-algebraic approach to quantum mechanics}}}\ (\bibinfo
  {publisher} {Springer},\ \bibinfo {year} {1975})\ pp.\ \bibinfo {pages}
  {293--328}\BibitemShut {NoStop}%
\bibitem [{\citenamefont {Raussendorf}(2013)}]{raussendorf2013contextuality}%
  \BibitemOpen
  \bibfield  {author} {\bibinfo {author} {\bibfnamefont {R.}~\bibnamefont
  {Raussendorf}},\ }\bibfield  {title} {\bibinfo {title} {Contextuality in
  measurement-based quantum computation},\ }\href
  {https://doi.org/10.1103/PhysRevA.88.022322} {\bibfield  {journal} {\bibinfo
  {journal} {Physical Review A}\ }\textbf {\bibinfo {volume} {88}},\ \bibinfo
  {pages} {022322} (\bibinfo {year} {2013})}\BibitemShut {NoStop}%
\bibitem [{\citenamefont {Howard}\ \emph {et~al.}(2014)\citenamefont {Howard},
  \citenamefont {Wallman}, \citenamefont {Veitch},\ and\ \citenamefont
  {Emerson}}]{howard2014contextuality}%
  \BibitemOpen
  \bibfield  {author} {\bibinfo {author} {\bibfnamefont {M.}~\bibnamefont
  {Howard}}, \bibinfo {author} {\bibfnamefont {J.}~\bibnamefont {Wallman}},
  \bibinfo {author} {\bibfnamefont {V.}~\bibnamefont {Veitch}},\ and\ \bibinfo
  {author} {\bibfnamefont {J.}~\bibnamefont {Emerson}},\ }\bibfield  {title}
  {\bibinfo {title} {Contextuality supplies the {`magic'} for quantum
  computation},\ }\href {https://doi.org/10.1038/nature13460} {\bibfield
  {journal} {\bibinfo  {journal} {Nature}\ }\textbf {\bibinfo {volume} {510}},\
  \bibinfo {pages} {351} (\bibinfo {year} {2014})}\BibitemShut {NoStop}%
\bibitem [{\citenamefont {Abramsky}\ \emph
  {et~al.}(2017{\natexlab{a}})\citenamefont {Abramsky}, \citenamefont
  {Barbosa},\ and\ \citenamefont {Mansfield}}]{abramsky2017contextual}%
  \BibitemOpen
  \bibfield  {author} {\bibinfo {author} {\bibfnamefont {S.}~\bibnamefont
  {Abramsky}}, \bibinfo {author} {\bibfnamefont {R.~S.}\ \bibnamefont
  {Barbosa}},\ and\ \bibinfo {author} {\bibfnamefont {S.}~\bibnamefont
  {Mansfield}},\ }\bibfield  {title} {\bibinfo {title} {Contextual fraction as
  a measure of contextuality},\ }\href
  {https://doi.org/10.1103/PhysRevLett.119.050504} {\bibfield  {journal}
  {\bibinfo  {journal} {Physical Review Letters}\ }\textbf {\bibinfo {volume}
  {119}},\ \bibinfo {pages} {050504} (\bibinfo {year}
  {2017}{\natexlab{a}})}\BibitemShut {NoStop}%
\bibitem [{\citenamefont {Bermejo-Vega}\ \emph {et~al.}(2017)\citenamefont
  {Bermejo-Vega}, \citenamefont {Delfosse}, \citenamefont {Browne},
  \citenamefont {Okay},\ and\ \citenamefont
  {Raussendorf}}]{bermejo2017contextuality}%
  \BibitemOpen
  \bibfield  {author} {\bibinfo {author} {\bibfnamefont {J.}~\bibnamefont
  {Bermejo-Vega}}, \bibinfo {author} {\bibfnamefont {N.}~\bibnamefont
  {Delfosse}}, \bibinfo {author} {\bibfnamefont {D.~E.}\ \bibnamefont
  {Browne}}, \bibinfo {author} {\bibfnamefont {C.}~\bibnamefont {Okay}},\ and\
  \bibinfo {author} {\bibfnamefont {R.}~\bibnamefont {Raussendorf}},\
  }\bibfield  {title} {\bibinfo {title} {Contextuality as a resource for models
  of quantum computation with qubits},\ }\href
  {https://doi.org/10.1103/PhysRevLett.119.120505} {\bibfield  {journal}
  {\bibinfo  {journal} {Physical Review Letters}\ }\textbf {\bibinfo {volume}
  {119}},\ \bibinfo {pages} {120505} (\bibinfo {year} {2017})}\BibitemShut
  {NoStop}%
\bibitem [{\citenamefont {Abramsky}\ \emph
  {et~al.}(2017{\natexlab{b}})\citenamefont {Abramsky}, \citenamefont
  {Barbosa}, \citenamefont {de~Silva},\ and\ \citenamefont
  {Zapata}}]{abramsky2017quantum}%
  \BibitemOpen
  \bibfield  {author} {\bibinfo {author} {\bibfnamefont {S.}~\bibnamefont
  {Abramsky}}, \bibinfo {author} {\bibfnamefont {R.~S.}\ \bibnamefont
  {Barbosa}}, \bibinfo {author} {\bibfnamefont {N.}~\bibnamefont {de~Silva}},\
  and\ \bibinfo {author} {\bibfnamefont {O.}~\bibnamefont {Zapata}},\
  }\bibfield  {title} {\bibinfo {title} {The quantum monad on relational
  structures},\ }in\ \href {https://doi.org/10.4230/LIPIcs.MFCS.2017.35} {\emph
  {\bibinfo {booktitle} {42nd International Symposium on Mathematical
  Foundations of Computer Science (MFCS 2017)}}},\ \bibinfo {series} {Leibniz
  International Proceedings in Informatics (LIPIcs)}, Vol.~\bibinfo {volume}
  {83},\ \bibinfo {editor} {edited by\ \bibinfo {editor} {\bibfnamefont
  {K.~G.}\ \bibnamefont {Larsen}}, \bibinfo {editor} {\bibfnamefont {H.~L.}\
  \bibnamefont {Bodlaender}},\ and\ \bibinfo {editor} {\bibfnamefont {J.-F.}\
  \bibnamefont {Raskin}}}\ (\bibinfo  {publisher} {Schloss
  Dagstuhl--Leibniz-Zentrum fuer Informatik},\ \bibinfo {year} {2017})\ pp.\
  \bibinfo {pages} {35:1--35:19}\BibitemShut {NoStop}%
\bibitem [{\citenamefont {Mari}\ and\ \citenamefont {Eisert}(2012)}]{Mari2012}%
  \BibitemOpen
  \bibfield  {author} {\bibinfo {author} {\bibfnamefont {A.}~\bibnamefont
  {Mari}}\ and\ \bibinfo {author} {\bibfnamefont {J.}~\bibnamefont {Eisert}},\
  }\bibfield  {title} {\bibinfo {title} {Positive {W}igner functions render
  classical simulation of quantum computation efficient},\ }\href
  {https://doi.org/10.1103/PhysRevLett.109.230503} {\bibfield  {journal}
  {\bibinfo  {journal} {Physical Review Letters}\ }\textbf {\bibinfo {volume}
  {109}},\ \bibinfo {pages} {230503} (\bibinfo {year} {2012})}\BibitemShut
  {NoStop}%
\bibitem [{\citenamefont {Lloyd}\ and\ \citenamefont
  {Braunstein}(1999)}]{lloyd1999quantum}%
  \BibitemOpen
  \bibfield  {author} {\bibinfo {author} {\bibfnamefont {S.}~\bibnamefont
  {Lloyd}}\ and\ \bibinfo {author} {\bibfnamefont {S.~L.}\ \bibnamefont
  {Braunstein}},\ }\bibfield  {title} {\bibinfo {title} {Quantum computation
  over continuous variables},\ }in\ \href
  {https://doi.org/10.1007/978-94-015-1258-9_2} {\emph {\bibinfo {booktitle}
  {Quantum Information with Continuous Variables}}}\ (\bibinfo  {publisher}
  {Springer},\ \bibinfo {year} {1999})\ pp.\ \bibinfo {pages}
  {9--17}\BibitemShut {NoStop}%
\bibitem [{\citenamefont {Abramsky}\ and\ \citenamefont
  {Brandenburger}(2011)}]{abramsky2011sheaf}%
  \BibitemOpen
  \bibfield  {author} {\bibinfo {author} {\bibfnamefont {S.}~\bibnamefont
  {Abramsky}}\ and\ \bibinfo {author} {\bibfnamefont {A.}~\bibnamefont
  {Brandenburger}},\ }\bibfield  {title} {\bibinfo {title} {The sheaf-theoretic
  structure of non-locality and contextuality},\ }\href
  {https://doi.org/10.1088/1367-2630/13/11/113036} {\bibfield  {journal}
  {\bibinfo  {journal} {New Journal of Physics}\ }\textbf {\bibinfo {volume}
  {13}},\ \bibinfo {pages} {113036} (\bibinfo {year} {2011})}\BibitemShut
  {NoStop}%
\bibitem [{\citenamefont {Cabello}\ \emph {et~al.}(2014)\citenamefont
  {Cabello}, \citenamefont {Severini},\ and\ \citenamefont {Winter}}]{csw}%
  \BibitemOpen
  \bibfield  {author} {\bibinfo {author} {\bibfnamefont {A.}~\bibnamefont
  {Cabello}}, \bibinfo {author} {\bibfnamefont {S.}~\bibnamefont {Severini}},\
  and\ \bibinfo {author} {\bibfnamefont {A.}~\bibnamefont {Winter}},\
  }\bibfield  {title} {\bibinfo {title} {Graph-theoretic approach to quantum
  correlations},\ }\href {https://doi.org/10.1103/PhysRevLett.112.040401}
  {\bibfield  {journal} {\bibinfo  {journal} {Physical Review Letters}\
  }\textbf {\bibinfo {volume} {112}},\ \bibinfo {pages} {040401} (\bibinfo
  {year} {2014})}\BibitemShut {NoStop}%
\bibitem [{\citenamefont {Spekkens}(2005)}]{Spekkens2005}%
  \BibitemOpen
  \bibfield  {author} {\bibinfo {author} {\bibfnamefont {R.~W.}\ \bibnamefont
  {Spekkens}},\ }\bibfield  {title} {\bibinfo {title} {Contextuality for
  preparations, transformations, and unsharp measurements},\ }\href
  {https://doi.org/10.1103/PhysRevA.71.052108} {\bibfield  {journal} {\bibinfo
  {journal} {Physical Review A}\ }\textbf {\bibinfo {volume} {71}},\ \bibinfo
  {pages} {052108} (\bibinfo {year} {2005})}\BibitemShut {NoStop}%
\bibitem [{\citenamefont {Zhang}\ \emph {et~al.}(2013)\citenamefont {Zhang},
  \citenamefont {Um}, \citenamefont {Zhang}, \citenamefont {An}, \citenamefont
  {Wang}, \citenamefont {Deng}, \citenamefont {Shen}, \citenamefont {Duan},\
  and\ \citenamefont {Kim}}]{Xiang2013}%
  \BibitemOpen
  \bibfield  {author} {\bibinfo {author} {\bibfnamefont {X.}~\bibnamefont
  {Zhang}}, \bibinfo {author} {\bibfnamefont {M.}~\bibnamefont {Um}}, \bibinfo
  {author} {\bibfnamefont {J.}~\bibnamefont {Zhang}}, \bibinfo {author}
  {\bibfnamefont {S.}~\bibnamefont {An}}, \bibinfo {author} {\bibfnamefont
  {Y.}~\bibnamefont {Wang}}, \bibinfo {author} {\bibfnamefont {D.-l.}\
  \bibnamefont {Deng}}, \bibinfo {author} {\bibfnamefont {C.}~\bibnamefont
  {Shen}}, \bibinfo {author} {\bibfnamefont {L.-M.}\ \bibnamefont {Duan}},\
  and\ \bibinfo {author} {\bibfnamefont {K.}~\bibnamefont {Kim}},\ }\bibfield
  {title} {\bibinfo {title} {State-independent experimental test of quantum
  contextuality with a single trapped ion},\ }\href
  {https://doi.org/10.1103/PhysRevLett.110.070401} {\bibfield  {journal}
  {\bibinfo  {journal} {Physical Review Letters}\ }\textbf {\bibinfo {volume}
  {110}},\ \bibinfo {pages} {070401} (\bibinfo {year} {2013})}\BibitemShut
  {NoStop}%
\bibitem [{\citenamefont {Hasegawa}\ \emph {et~al.}(2006)\citenamefont
  {Hasegawa}, \citenamefont {Loidl}, \citenamefont {Badurek}, \citenamefont
  {Baron},\ and\ \citenamefont {Rauch}}]{Helmut2006}%
  \BibitemOpen
  \bibfield  {author} {\bibinfo {author} {\bibfnamefont {Y.}~\bibnamefont
  {Hasegawa}}, \bibinfo {author} {\bibfnamefont {R.}~\bibnamefont {Loidl}},
  \bibinfo {author} {\bibfnamefont {G.}~\bibnamefont {Badurek}}, \bibinfo
  {author} {\bibfnamefont {M.}~\bibnamefont {Baron}},\ and\ \bibinfo {author}
  {\bibfnamefont {H.}~\bibnamefont {Rauch}},\ }\bibfield  {title} {\bibinfo
  {title} {Quantum contextuality in a single-neutron optical experiment},\
  }\href {https://doi.org/10.1103/PhysRevLett.97.230401} {\bibfield  {journal}
  {\bibinfo  {journal} {Physical Review Letters}\ }\textbf {\bibinfo {volume}
  {97}},\ \bibinfo {pages} {230401} (\bibinfo {year} {2006})}\BibitemShut
  {NoStop}%
\bibitem [{\citenamefont {Bartosik}\ \emph {et~al.}(2009)\citenamefont
  {Bartosik}, \citenamefont {Klepp}, \citenamefont {Schmitzer}, \citenamefont
  {Sponar}, \citenamefont {Cabello}, \citenamefont {Rauch},\ and\ \citenamefont
  {Hasegawa}}]{bartosik2009experimental}%
  \BibitemOpen
  \bibfield  {author} {\bibinfo {author} {\bibfnamefont {H.}~\bibnamefont
  {Bartosik}}, \bibinfo {author} {\bibfnamefont {J.}~\bibnamefont {Klepp}},
  \bibinfo {author} {\bibfnamefont {C.}~\bibnamefont {Schmitzer}}, \bibinfo
  {author} {\bibfnamefont {S.}~\bibnamefont {Sponar}}, \bibinfo {author}
  {\bibfnamefont {A.}~\bibnamefont {Cabello}}, \bibinfo {author} {\bibfnamefont
  {H.}~\bibnamefont {Rauch}},\ and\ \bibinfo {author} {\bibfnamefont
  {Y.}~\bibnamefont {Hasegawa}},\ }\bibfield  {title} {\bibinfo {title}
  {Experimental test of quantum contextuality in neutron interferometry},\
  }\href {https://doi.org/10.1103/PhysRevLett.103.040403} {\bibfield  {journal}
  {\bibinfo  {journal} {Physical Review Letters}\ }\textbf {\bibinfo {volume}
  {103}},\ \bibinfo {pages} {040403} (\bibinfo {year} {2009})}\BibitemShut
  {NoStop}%
\bibitem [{\citenamefont {Kirchmair}\ \emph {et~al.}(2009)\citenamefont
  {Kirchmair}, \citenamefont {Z{\"a}hringer}, \citenamefont {Gerritsma},
  \citenamefont {Kleinmann}, \citenamefont {G{\"u}hne}, \citenamefont
  {Cabello}, \citenamefont {Blatt},\ and\ \citenamefont
  {Roos}}]{kirchmair2009state}%
  \BibitemOpen
  \bibfield  {author} {\bibinfo {author} {\bibfnamefont {G.}~\bibnamefont
  {Kirchmair}}, \bibinfo {author} {\bibfnamefont {F.}~\bibnamefont
  {Z{\"a}hringer}}, \bibinfo {author} {\bibfnamefont {R.}~\bibnamefont
  {Gerritsma}}, \bibinfo {author} {\bibfnamefont {M.}~\bibnamefont
  {Kleinmann}}, \bibinfo {author} {\bibfnamefont {O.}~\bibnamefont
  {G{\"u}hne}}, \bibinfo {author} {\bibfnamefont {A.}~\bibnamefont {Cabello}},
  \bibinfo {author} {\bibfnamefont {R.}~\bibnamefont {Blatt}},\ and\ \bibinfo
  {author} {\bibfnamefont {C.~F.}\ \bibnamefont {Roos}},\ }\bibfield  {title}
  {\bibinfo {title} {State-independent experimental test of quantum
  contextuality},\ }\href {https://doi.org/10.1038/nature08172} {\bibfield
  {journal} {\bibinfo  {journal} {Nature}\ }\textbf {\bibinfo {volume} {460}},\
  \bibinfo {pages} {494} (\bibinfo {year} {2009})}\BibitemShut {NoStop}%
\bibitem [{\citenamefont {Spekkens}(2008)}]{spekkens2008negativity}%
  \BibitemOpen
  \bibfield  {author} {\bibinfo {author} {\bibfnamefont {R.~W.}\ \bibnamefont
  {Spekkens}},\ }\bibfield  {title} {\bibinfo {title} {Negativity and
  contextuality are equivalent notions of nonclassicality},\ }\href
  {https://doi.org/10.1103/PhysRevLett.101.020401} {\bibfield  {journal}
  {\bibinfo  {journal} {Physical Review Letters}\ }\textbf {\bibinfo {volume}
  {101}},\ \bibinfo {pages} {020401} (\bibinfo {year} {2008})}\BibitemShut
  {NoStop}%
\bibitem [{\citenamefont {Gross}(2006)}]{gross2006hudson}%
  \BibitemOpen
  \bibfield  {author} {\bibinfo {author} {\bibfnamefont {D.}~\bibnamefont
  {Gross}},\ }\bibfield  {title} {\bibinfo {title} {Hudson’s theorem for
  finite-dimensional quantum systems},\ }\href
  {https://doi.org/10.1063/1.2393152} {\bibfield  {journal} {\bibinfo
  {journal} {Journal of mathematical physics}\ }\textbf {\bibinfo {volume}
  {47}},\ \bibinfo {pages} {122107} (\bibinfo {year} {2006})}\BibitemShut
  {NoStop}%
\bibitem [{\citenamefont {Delfosse}\ \emph {et~al.}(2017)\citenamefont
  {Delfosse}, \citenamefont {Okay}, \citenamefont {Bermejo-Vega}, \citenamefont
  {Browne},\ and\ \citenamefont {Raussendorf}}]{delfosse2017equivalence}%
  \BibitemOpen
  \bibfield  {author} {\bibinfo {author} {\bibfnamefont {N.}~\bibnamefont
  {Delfosse}}, \bibinfo {author} {\bibfnamefont {C.}~\bibnamefont {Okay}},
  \bibinfo {author} {\bibfnamefont {J.}~\bibnamefont {Bermejo-Vega}}, \bibinfo
  {author} {\bibfnamefont {D.~E.}\ \bibnamefont {Browne}},\ and\ \bibinfo
  {author} {\bibfnamefont {R.}~\bibnamefont {Raussendorf}},\ }\bibfield
  {title} {\bibinfo {title} {Equivalence between contextuality and negativity
  of the {W}igner function for qudits},\ }\href
  {https://doi.org/10.1088/1367-2630/aa8fe3} {\bibfield  {journal} {\bibinfo
  {journal} {New Journal of Physics}\ }\textbf {\bibinfo {volume} {19}},\
  \bibinfo {pages} {123024} (\bibinfo {year} {2017})}\BibitemShut {NoStop}%
\bibitem [{\citenamefont {Raussendorf}\ \emph {et~al.}(2017)\citenamefont
  {Raussendorf}, \citenamefont {Browne}, \citenamefont {Delfosse},
  \citenamefont {Okay},\ and\ \citenamefont
  {Bermejo-Vega}}]{raussendorf2017contextuality}%
  \BibitemOpen
  \bibfield  {author} {\bibinfo {author} {\bibfnamefont {R.}~\bibnamefont
  {Raussendorf}}, \bibinfo {author} {\bibfnamefont {D.~E.}\ \bibnamefont
  {Browne}}, \bibinfo {author} {\bibfnamefont {N.}~\bibnamefont {Delfosse}},
  \bibinfo {author} {\bibfnamefont {C.}~\bibnamefont {Okay}},\ and\ \bibinfo
  {author} {\bibfnamefont {J.}~\bibnamefont {Bermejo-Vega}},\ }\bibfield
  {title} {\bibinfo {title} {Contextuality and {W}igner-function negativity in
  qubit quantum computation},\ }\href
  {https://doi.org/10.1103/PhysRevA.95.052334} {\bibfield  {journal} {\bibinfo
  {journal} {Physical Review A}\ }\textbf {\bibinfo {volume} {95}},\ \bibinfo
  {pages} {052334} (\bibinfo {year} {2017})}\BibitemShut {NoStop}%
\bibitem [{\citenamefont {Delfosse}\ \emph {et~al.}(2015)\citenamefont
  {Delfosse}, \citenamefont {Allard~Guerin}, \citenamefont {Bian},\ and\
  \citenamefont {Raussendorf}}]{delfosse2015wigner}%
  \BibitemOpen
  \bibfield  {author} {\bibinfo {author} {\bibfnamefont {N.}~\bibnamefont
  {Delfosse}}, \bibinfo {author} {\bibfnamefont {P.}~\bibnamefont
  {Allard~Guerin}}, \bibinfo {author} {\bibfnamefont {J.}~\bibnamefont
  {Bian}},\ and\ \bibinfo {author} {\bibfnamefont {R.}~\bibnamefont
  {Raussendorf}},\ }\bibfield  {title} {\bibinfo {title} {{W}igner function
  negativity and contextuality in quantum computation on rebits},\ }\href
  {https://doi.org/10.1103/PhysRevX.5.021003} {\bibfield  {journal} {\bibinfo
  {journal} {Physical Review X}\ }\textbf {\bibinfo {volume} {5}},\ \bibinfo
  {pages} {021003} (\bibinfo {year} {2015})}\BibitemShut {NoStop}%
\bibitem [{\citenamefont {Schmid}\ \emph {et~al.}(2021)\citenamefont {Schmid},
  \citenamefont {Du}, \citenamefont {Selby},\ and\ \citenamefont
  {Pusey}}]{schmid2021only}%
  \BibitemOpen
  \bibfield  {author} {\bibinfo {author} {\bibfnamefont {D.}~\bibnamefont
  {Schmid}}, \bibinfo {author} {\bibfnamefont {H.}~\bibnamefont {Du}}, \bibinfo
  {author} {\bibfnamefont {J.~H.}\ \bibnamefont {Selby}},\ and\ \bibinfo
  {author} {\bibfnamefont {M.~F.}\ \bibnamefont {Pusey}},\ }\href@noop {}
  {\bibinfo {title} {The only noncontextual model of the stabilizer subtheory
  is gross's}},\ \bibinfo {howpublished} {Preprint
  \href{https://arxiv.org/abs/2101.06263}{arXiv:2101.06263 [quant-ph]}}
  (\bibinfo {year} {2021})\BibitemShut {NoStop}%
\bibitem [{\citenamefont {Braunstein}\ and\ \citenamefont
  {Van~Loock}(2005)}]{braunstein2005quantum}%
  \BibitemOpen
  \bibfield  {author} {\bibinfo {author} {\bibfnamefont {S.~L.}\ \bibnamefont
  {Braunstein}}\ and\ \bibinfo {author} {\bibfnamefont {P.}~\bibnamefont
  {Van~Loock}},\ }\bibfield  {title} {\bibinfo {title} {Quantum information
  with continuous variables},\ }\href
  {https://doi.org/10.1103/RevModPhys.77.513} {\bibfield  {journal} {\bibinfo
  {journal} {Reviews of modern physics}\ }\textbf {\bibinfo {volume} {77}},\
  \bibinfo {pages} {513} (\bibinfo {year} {2005})}\BibitemShut {NoStop}%
\bibitem [{\citenamefont {Weedbrook}\ \emph {et~al.}(2012)\citenamefont
  {Weedbrook}, \citenamefont {Pirandola}, \citenamefont
  {Garc{\'\i}a-Patr{\'o}n}, \citenamefont {Cerf}, \citenamefont {Ralph},
  \citenamefont {Shapiro},\ and\ \citenamefont
  {Lloyd}}]{weedbrook2012gaussian}%
  \BibitemOpen
  \bibfield  {author} {\bibinfo {author} {\bibfnamefont {C.}~\bibnamefont
  {Weedbrook}}, \bibinfo {author} {\bibfnamefont {S.}~\bibnamefont
  {Pirandola}}, \bibinfo {author} {\bibfnamefont {R.}~\bibnamefont
  {Garc{\'\i}a-Patr{\'o}n}}, \bibinfo {author} {\bibfnamefont {N.~J.}\
  \bibnamefont {Cerf}}, \bibinfo {author} {\bibfnamefont {T.~C.}\ \bibnamefont
  {Ralph}}, \bibinfo {author} {\bibfnamefont {J.~H.}\ \bibnamefont {Shapiro}},\
  and\ \bibinfo {author} {\bibfnamefont {S.}~\bibnamefont {Lloyd}},\ }\bibfield
   {title} {\bibinfo {title} {{G}aussian quantum information},\ }\href
  {https://doi.org/10.1103/RevModPhys.84.621} {\bibfield  {journal} {\bibinfo
  {journal} {Reviews of Modern Physics}\ }\textbf {\bibinfo {volume} {84}},\
  \bibinfo {pages} {621} (\bibinfo {year} {2012})}\BibitemShut {NoStop}%
\bibitem [{\citenamefont {Crespi}\ \emph {et~al.}(2013)\citenamefont {Crespi},
  \citenamefont {Osellame}, \citenamefont {Ramponi}, \citenamefont {Brod},
  \citenamefont {Galvao}, \citenamefont {Spagnolo}, \citenamefont {Vitelli},
  \citenamefont {Maiorino}, \citenamefont {Mataloni},\ and\ \citenamefont
  {Sciarrino}}]{crespi2013integrated}%
  \BibitemOpen
  \bibfield  {author} {\bibinfo {author} {\bibfnamefont {A.}~\bibnamefont
  {Crespi}}, \bibinfo {author} {\bibfnamefont {R.}~\bibnamefont {Osellame}},
  \bibinfo {author} {\bibfnamefont {R.}~\bibnamefont {Ramponi}}, \bibinfo
  {author} {\bibfnamefont {D.~J.}\ \bibnamefont {Brod}}, \bibinfo {author}
  {\bibfnamefont {E.~F.}\ \bibnamefont {Galvao}}, \bibinfo {author}
  {\bibfnamefont {N.}~\bibnamefont {Spagnolo}}, \bibinfo {author}
  {\bibfnamefont {C.}~\bibnamefont {Vitelli}}, \bibinfo {author} {\bibfnamefont
  {E.}~\bibnamefont {Maiorino}}, \bibinfo {author} {\bibfnamefont
  {P.}~\bibnamefont {Mataloni}},\ and\ \bibinfo {author} {\bibfnamefont
  {F.}~\bibnamefont {Sciarrino}},\ }\bibfield  {title} {\bibinfo {title}
  {Integrated multimode interferometers with arbitrary designs for photonic
  boson sampling},\ }\href {https://doi.org/10.1038/nphoton.2013.112}
  {\bibfield  {journal} {\bibinfo  {journal} {Nature photonics}\ }\textbf
  {\bibinfo {volume} {7}},\ \bibinfo {pages} {545} (\bibinfo {year}
  {2013})}\BibitemShut {NoStop}%
\bibitem [{\citenamefont {Bourassa}\ \emph {et~al.}(2021)\citenamefont
  {Bourassa}, \citenamefont {Alexander}, \citenamefont {Vasmer}, \citenamefont
  {Patil}, \citenamefont {Tzitrin}, \citenamefont {Matsuura}, \citenamefont
  {Su}, \citenamefont {Baragiola}, \citenamefont {Guha}, \citenamefont
  {Dauphinais}, \citenamefont {Sabapathy}, \citenamefont {Menicucci},\ and\
  \citenamefont {Dhand}}]{Bourassa2021blueprintscalable}%
  \BibitemOpen
  \bibfield  {author} {\bibinfo {author} {\bibfnamefont {J.~E.}\ \bibnamefont
  {Bourassa}}, \bibinfo {author} {\bibfnamefont {R.~N.}\ \bibnamefont
  {Alexander}}, \bibinfo {author} {\bibfnamefont {M.}~\bibnamefont {Vasmer}},
  \bibinfo {author} {\bibfnamefont {A.}~\bibnamefont {Patil}}, \bibinfo
  {author} {\bibfnamefont {I.}~\bibnamefont {Tzitrin}}, \bibinfo {author}
  {\bibfnamefont {T.}~\bibnamefont {Matsuura}}, \bibinfo {author}
  {\bibfnamefont {D.}~\bibnamefont {Su}}, \bibinfo {author} {\bibfnamefont
  {B.~Q.}\ \bibnamefont {Baragiola}}, \bibinfo {author} {\bibfnamefont
  {S.}~\bibnamefont {Guha}}, \bibinfo {author} {\bibfnamefont {G.}~\bibnamefont
  {Dauphinais}}, \bibinfo {author} {\bibfnamefont {K.~K.}\ \bibnamefont
  {Sabapathy}}, \bibinfo {author} {\bibfnamefont {N.~C.}\ \bibnamefont
  {Menicucci}},\ and\ \bibinfo {author} {\bibfnamefont {I.}~\bibnamefont
  {Dhand}},\ }\bibfield  {title} {\bibinfo {title} {Blueprint for a {S}calable
  {P}hotonic {F}ault-{T}olerant {Q}uantum {C}omputer},\ }\href
  {https://doi.org/10.22331/q-2021-02-04-392} {\bibfield  {journal} {\bibinfo
  {journal} {{Quantum}}\ }\textbf {\bibinfo {volume} {5}},\ \bibinfo {pages}
  {392} (\bibinfo {year} {2021})}\BibitemShut {NoStop}%
\bibitem [{\citenamefont {Walschaers}(2021)}]{walschaers2021non}%
  \BibitemOpen
  \bibfield  {author} {\bibinfo {author} {\bibfnamefont {M.}~\bibnamefont
  {Walschaers}},\ }\href@noop {} {\bibinfo {title} {Non-{G}aussian quantum
  states and where to find them}},\ \bibinfo {howpublished} {Preprint
  \href{https://arxiv.org/abs/2104.12596}{arXiv:2104.12596 [quant-ph]}}
  (\bibinfo {year} {2021})\BibitemShut {NoStop}%
\bibitem [{\citenamefont {Chabaud}(2021)}]{chabaud2021continuous}%
  \BibitemOpen
  \bibfield  {author} {\bibinfo {author} {\bibfnamefont {U.}~\bibnamefont
  {Chabaud}},\ }\href@noop {} {\bibinfo {title} {Continuous variable quantum
  advantages and applications in quantum optics}},\ \bibinfo {howpublished}
  {PhD thesis \href{https://arxiv.org/abs/2102.05227}{arXiv:2102.05227
  [quant-ph]}} (\bibinfo {year} {Sorbonne University, 2021})\BibitemShut
  {NoStop}%
\bibitem [{\citenamefont {Gottesman}\ \emph
  {et~al.}(2001{\natexlab{a}})\citenamefont {Gottesman}, \citenamefont
  {Kitaev},\ and\ \citenamefont {Preskill}}]{Gottesman2001}%
  \BibitemOpen
  \bibfield  {author} {\bibinfo {author} {\bibfnamefont {D.}~\bibnamefont
  {Gottesman}}, \bibinfo {author} {\bibfnamefont {A.}~\bibnamefont {Kitaev}},\
  and\ \bibinfo {author} {\bibfnamefont {J.}~\bibnamefont {Preskill}},\
  }\bibfield  {title} {\bibinfo {title} {Encoding a qubit in an oscillator},\
  }\href {https://doi.org/10.1103/PhysRevA.64.012310} {\bibfield  {journal}
  {\bibinfo  {journal} {Phys. Rev. A}\ }\textbf {\bibinfo {volume} {64}},\
  \bibinfo {pages} {012310} (\bibinfo {year} {2001}{\natexlab{a}})}\BibitemShut
  {NoStop}%
\bibitem [{\citenamefont {Cai}\ \emph {et~al.}(2021)\citenamefont {Cai},
  \citenamefont {Ma}, \citenamefont {Wang}, \citenamefont {Zou},\ and\
  \citenamefont {Sun}}]{cai2021bosonic}%
  \BibitemOpen
  \bibfield  {author} {\bibinfo {author} {\bibfnamefont {W.}~\bibnamefont
  {Cai}}, \bibinfo {author} {\bibfnamefont {Y.}~\bibnamefont {Ma}}, \bibinfo
  {author} {\bibfnamefont {W.}~\bibnamefont {Wang}}, \bibinfo {author}
  {\bibfnamefont {C.-L.}\ \bibnamefont {Zou}},\ and\ \bibinfo {author}
  {\bibfnamefont {L.}~\bibnamefont {Sun}},\ }\bibfield  {title} {\bibinfo
  {title} {Bosonic quantum error correction codes in superconducting quantum
  circuits},\ }\href {https://doi.org/10.1016/j.fmre.2020.12.006} {\bibfield
  {journal} {\bibinfo  {journal} {Fundamental Research}\ }\textbf {\bibinfo
  {volume} {1}},\ \bibinfo {pages} {50} (\bibinfo {year} {2021})}\BibitemShut
  {NoStop}%
\bibitem [{\citenamefont {Yokoyama}\ \emph {et~al.}(2013)\citenamefont
  {Yokoyama}, \citenamefont {Ukai}, \citenamefont {Armstrong}, \citenamefont
  {Sornphiphatphong}, \citenamefont {Kaji}, \citenamefont {Suzuki},
  \citenamefont {Yoshikawa}, \citenamefont {Yonezawa}, \citenamefont
  {Menicucci},\ and\ \citenamefont {Furusawa}}]{yokoyama2013ultra}%
  \BibitemOpen
  \bibfield  {author} {\bibinfo {author} {\bibfnamefont {S.}~\bibnamefont
  {Yokoyama}}, \bibinfo {author} {\bibfnamefont {R.}~\bibnamefont {Ukai}},
  \bibinfo {author} {\bibfnamefont {S.~C.}\ \bibnamefont {Armstrong}}, \bibinfo
  {author} {\bibfnamefont {C.}~\bibnamefont {Sornphiphatphong}}, \bibinfo
  {author} {\bibfnamefont {T.}~\bibnamefont {Kaji}}, \bibinfo {author}
  {\bibfnamefont {S.}~\bibnamefont {Suzuki}}, \bibinfo {author} {\bibfnamefont
  {J.-i.}\ \bibnamefont {Yoshikawa}}, \bibinfo {author} {\bibfnamefont
  {H.}~\bibnamefont {Yonezawa}}, \bibinfo {author} {\bibfnamefont {N.~C.}\
  \bibnamefont {Menicucci}},\ and\ \bibinfo {author} {\bibfnamefont
  {A.}~\bibnamefont {Furusawa}},\ }\bibfield  {title} {\bibinfo {title}
  {Ultra-large-scale continuous-variable cluster states multiplexed in the time
  domain},\ }\href {https://doi.org/10.1038/nphoton.2013.287} {\bibfield
  {journal} {\bibinfo  {journal} {Nature Photonics}\ }\textbf {\bibinfo
  {volume} {7}},\ \bibinfo {pages} {982} (\bibinfo {year} {2013})}\BibitemShut
  {NoStop}%
\bibitem [{\citenamefont {Yoshikawa}\ \emph {et~al.}(2016)\citenamefont
  {Yoshikawa}, \citenamefont {Yokoyama}, \citenamefont {Kaji}, \citenamefont
  {Sornphiphatphong}, \citenamefont {Shiozawa}, \citenamefont {Makino},\ and\
  \citenamefont {Furusawa}}]{yoshikawa2016invited}%
  \BibitemOpen
  \bibfield  {author} {\bibinfo {author} {\bibfnamefont {J.-i.}\ \bibnamefont
  {Yoshikawa}}, \bibinfo {author} {\bibfnamefont {S.}~\bibnamefont {Yokoyama}},
  \bibinfo {author} {\bibfnamefont {T.}~\bibnamefont {Kaji}}, \bibinfo {author}
  {\bibfnamefont {C.}~\bibnamefont {Sornphiphatphong}}, \bibinfo {author}
  {\bibfnamefont {Y.}~\bibnamefont {Shiozawa}}, \bibinfo {author}
  {\bibfnamefont {K.}~\bibnamefont {Makino}},\ and\ \bibinfo {author}
  {\bibfnamefont {A.}~\bibnamefont {Furusawa}},\ }\bibfield  {title} {\bibinfo
  {title} {Invited article: Generation of one-million-mode continuous-variable
  cluster state by unlimited time-domain multiplexing},\ }\href
  {https://doi.org/10.1063/1.4962732} {\bibfield  {journal} {\bibinfo
  {journal} {APL Photonics}\ }\textbf {\bibinfo {volume} {1}},\ \bibinfo
  {pages} {060801} (\bibinfo {year} {2016})}\BibitemShut {NoStop}%
\bibitem [{\citenamefont {Leonhardt}(2010)}]{Leonhardt-essential}%
  \BibitemOpen
  \bibfield  {author} {\bibinfo {author} {\bibfnamefont {U.}~\bibnamefont
  {Leonhardt}},\ }\href {https://doi.org/10.1017/CBO9780511806117} {\emph
  {\bibinfo {title} {Essential Quantum Optics}}},\ \bibinfo {edition} {1st}\
  ed.\ (\bibinfo  {publisher} {Cambridge University Press},\ \bibinfo {address}
  {Cambridge, UK},\ \bibinfo {year} {2010})\BibitemShut {NoStop}%
\bibitem [{\citenamefont {Ohliger}\ and\ \citenamefont
  {Eisert}(2012)}]{Ohliger2012}%
  \BibitemOpen
  \bibfield  {author} {\bibinfo {author} {\bibfnamefont {M.}~\bibnamefont
  {Ohliger}}\ and\ \bibinfo {author} {\bibfnamefont {J.}~\bibnamefont
  {Eisert}},\ }\bibfield  {title} {\bibinfo {title} {Efficient
  measurement-based quantum computing with continuous-variable systems},\
  }\href {https://doi.org/10.1103/PhysRevA.85.062318} {\bibfield  {journal}
  {\bibinfo  {journal} {Physical Review A}\ }\textbf {\bibinfo {volume} {85}},\
  \bibinfo {pages} {062318} (\bibinfo {year} {2012})}\BibitemShut {NoStop}%
\bibitem [{\citenamefont {Barbosa}\ \emph {et~al.}(2022)\citenamefont
  {Barbosa}, \citenamefont {Douce}, \citenamefont {Emeriau}, \citenamefont
  {Kashefi},\ and\ \citenamefont {Mansfield}}]{barbosa2019continuous}%
  \BibitemOpen
  \bibfield  {author} {\bibinfo {author} {\bibfnamefont {R.~S.}\ \bibnamefont
  {Barbosa}}, \bibinfo {author} {\bibfnamefont {T.}~\bibnamefont {Douce}},
  \bibinfo {author} {\bibfnamefont {P.-E.}\ \bibnamefont {Emeriau}}, \bibinfo
  {author} {\bibfnamefont {E.}~\bibnamefont {Kashefi}},\ and\ \bibinfo {author}
  {\bibfnamefont {S.}~\bibnamefont {Mansfield}},\ }\bibfield  {title} {\bibinfo
  {title} {Continuous-variable nonlocality and contextuality},\ }\href
  {https://doi.org/10.1007/s00220-021-04285-7} {\bibfield  {journal} {\bibinfo
  {journal} {Communications in Mathematical Physics}\ }\textbf {\bibinfo
  {volume} {391}},\ \bibinfo {pages} {1047} (\bibinfo {year}
  {2022})}\BibitemShut {NoStop}%
\bibitem [{\citenamefont {Simon~Kochen}(1968)}]{KochenSpecker1967}%
  \BibitemOpen
  \bibfield  {author} {\bibinfo {author} {\bibfnamefont {E.~S.}\ \bibnamefont
  {Simon~Kochen}},\ }\bibfield  {title} {\bibinfo {title} {The problem of
  hidden variables in quantum mechanics},\ }\href@noop {} {\bibfield  {journal}
  {\bibinfo  {journal} {Indiana Univ. Math. J.}\ }\textbf {\bibinfo {volume}
  {17}},\ \bibinfo {pages} {59} (\bibinfo {year} {1968})}\BibitemShut {NoStop}%
\bibitem [{\citenamefont {Adesso}\ \emph {et~al.}(2014)\citenamefont {Adesso},
  \citenamefont {Ragy},\ and\ \citenamefont {Lee}}]{adesso2014continuous}%
  \BibitemOpen
  \bibfield  {author} {\bibinfo {author} {\bibfnamefont {G.}~\bibnamefont
  {Adesso}}, \bibinfo {author} {\bibfnamefont {S.}~\bibnamefont {Ragy}},\ and\
  \bibinfo {author} {\bibfnamefont {A.~R.}\ \bibnamefont {Lee}},\ }\bibfield
  {title} {\bibinfo {title} {Continuous variable quantum information:
  {G}aussian states and beyond},\ }\href
  {https://doi.org/10.1142/s1230161214400010} {\bibfield  {journal} {\bibinfo
  {journal} {Open Systems \& Information Dynamics}\ }\textbf {\bibinfo {volume}
  {21}},\ \bibinfo {pages} {1440001} (\bibinfo {year} {2014})}\BibitemShut
  {NoStop}%
\bibitem [{\citenamefont {Simon}\ \emph {et~al.}(1988)\citenamefont {Simon},
  \citenamefont {Sudarshan},\ and\ \citenamefont {Mukunda}}]{Sudarshan1988}%
  \BibitemOpen
  \bibfield  {author} {\bibinfo {author} {\bibfnamefont {R.}~\bibnamefont
  {Simon}}, \bibinfo {author} {\bibfnamefont {E.~C.~G.}\ \bibnamefont
  {Sudarshan}},\ and\ \bibinfo {author} {\bibfnamefont {N.}~\bibnamefont
  {Mukunda}},\ }\bibfield  {title} {\bibinfo {title} {Gaussian pure states in
  quantum mechanics and the symplectic group},\ }\href
  {https://doi.org/10.1103/PhysRevA.37.3028} {\bibfield  {journal} {\bibinfo
  {journal} {Physical Review A}\ }\textbf {\bibinfo {volume} {37}},\ \bibinfo
  {pages} {3028} (\bibinfo {year} {1988})}\BibitemShut {NoStop}%
\bibitem [{\citenamefont {De~Gosson}(2006)}]{Gosson2006symplectic}%
  \BibitemOpen
  \bibfield  {author} {\bibinfo {author} {\bibfnamefont {M.~A.}\ \bibnamefont
  {De~Gosson}},\ }\href {https://doi.org/10.1007/3-7643-7575-2} {\emph
  {\bibinfo {title} {Symplectic geometry and quantum mechanics}}},\ Vol.\
  \bibinfo {volume} {166}\ (\bibinfo  {publisher} {Springer Science \& Business
  Media},\ \bibinfo {year} {2006})\BibitemShut {NoStop}%
\bibitem [{\citenamefont {Ferraro}\ \emph {et~al.}(2005)\citenamefont
  {Ferraro}, \citenamefont {Olivares},\ and\ \citenamefont
  {Paris}}]{ferraro2005gaussian}%
  \BibitemOpen
  \bibfield  {author} {\bibinfo {author} {\bibfnamefont {A.}~\bibnamefont
  {Ferraro}}, \bibinfo {author} {\bibfnamefont {S.}~\bibnamefont {Olivares}},\
  and\ \bibinfo {author} {\bibfnamefont {M.~G.}\ \bibnamefont {Paris}},\
  }\href@noop {} {\bibinfo {title} {{G}aussian states in continuous variable
  quantum information}},\ \bibinfo {howpublished} {Preprint
  \href{https://arxiv.org/abs/quant-ph/0503237}{arXiv:0503237 [quant-ph]}}
  (\bibinfo {year} {2005})\BibitemShut {NoStop}%
\bibitem [{\citenamefont {Cahill}\ and\ \citenamefont
  {Glauber}(1969)}]{cahill1969density}%
  \BibitemOpen
  \bibfield  {author} {\bibinfo {author} {\bibfnamefont {K.~E.}\ \bibnamefont
  {Cahill}}\ and\ \bibinfo {author} {\bibfnamefont {R.~J.}\ \bibnamefont
  {Glauber}},\ }\bibfield  {title} {\bibinfo {title} {Density operators and
  quasiprobability distributions},\ }\href
  {https://doi.org/10.1103/PhysRev.177.1882} {\bibfield  {journal} {\bibinfo
  {journal} {Physical Review}\ }\textbf {\bibinfo {volume} {177}},\ \bibinfo
  {pages} {1882} (\bibinfo {year} {1969})}\BibitemShut {NoStop}%
\bibitem [{\citenamefont {{de Gosson}}(2017)}]{de_gosson_wigner_2017}%
  \BibitemOpen
  \bibfield  {author} {\bibinfo {author} {\bibfnamefont {M.}~\bibnamefont {{de
  Gosson}}},\ }\href {https://doi.org/10.1142/q0089} {\emph {\bibinfo {title}
  {The {{Wigner Transform}}}}}\ (\bibinfo  {publisher} {{World Scientific
  (Europe)}},\ \bibinfo {year} {2017})\BibitemShut {NoStop}%
\bibitem [{\citenamefont {{de Gosson}}(2011)}]{de_gosson_symplectic_2011}%
  \BibitemOpen
  \bibfield  {author} {\bibinfo {author} {\bibfnamefont {M.~A.}\ \bibnamefont
  {{de Gosson}}},\ }\href@noop {} {\emph {\bibinfo {title} {Symplectic
  {{Methods}} in {{Harmonic Analysis}} and in {{Mathematical Physics}}}}},\
  Pseudo-{{Differential Operators}}\ (\bibinfo  {publisher} {{Springer,
  Basel}},\ \bibinfo {year} {2011})\BibitemShut {NoStop}%
\bibitem [{\citenamefont {Billingsley}(2008)}]{billingsley2008probability}%
  \BibitemOpen
  \bibfield  {author} {\bibinfo {author} {\bibfnamefont {P.}~\bibnamefont
  {Billingsley}},\ }\href {https://doi.org/10.2307/2987970} {\emph {\bibinfo
  {title} {Probability and measure}}}\ (\bibinfo  {publisher} {John Wiley \&
  Sons},\ \bibinfo {year} {2008})\BibitemShut {NoStop}%
\bibitem [{\citenamefont {Tao}(2011)}]{tao2011introduction}%
  \BibitemOpen
  \bibfield  {author} {\bibinfo {author} {\bibfnamefont {T.}~\bibnamefont
  {Tao}},\ }\href {https://doi.org/10.1090/gsm/126} {\emph {\bibinfo {title}
  {An introduction to measure theory}}},\ Vol.\ \bibinfo {volume} {126}\
  (\bibinfo  {publisher} {American Mathematical Society Providence, RI},\
  \bibinfo {year} {2011})\BibitemShut {NoStop}%
\bibitem [{\citenamefont {Su}\ \emph {et~al.}(2013)\citenamefont {Su},
  \citenamefont {Hao}, \citenamefont {Deng}, \citenamefont {Ma}, \citenamefont
  {Wang}, \citenamefont {Jia}, \citenamefont {Xie},\ and\ \citenamefont
  {Peng}}]{su2013gate}%
  \BibitemOpen
  \bibfield  {author} {\bibinfo {author} {\bibfnamefont {X.}~\bibnamefont
  {Su}}, \bibinfo {author} {\bibfnamefont {S.}~\bibnamefont {Hao}}, \bibinfo
  {author} {\bibfnamefont {X.}~\bibnamefont {Deng}}, \bibinfo {author}
  {\bibfnamefont {L.}~\bibnamefont {Ma}}, \bibinfo {author} {\bibfnamefont
  {M.}~\bibnamefont {Wang}}, \bibinfo {author} {\bibfnamefont {X.}~\bibnamefont
  {Jia}}, \bibinfo {author} {\bibfnamefont {C.}~\bibnamefont {Xie}},\ and\
  \bibinfo {author} {\bibfnamefont {K.}~\bibnamefont {Peng}},\ }\bibfield
  {title} {\bibinfo {title} {Gate sequence for continuous variable one-way
  quantum computation},\ }\href {https://doi.org/10.1038/ncomms3828} {\bibfield
   {journal} {\bibinfo  {journal} {Nature communications}\ }\textbf {\bibinfo
  {volume} {4}},\ \bibinfo {pages} {1} (\bibinfo {year} {2013})}\BibitemShut
  {NoStop}%
\bibitem [{\citenamefont {Douce}\ \emph {et~al.}(2017)\citenamefont {Douce},
  \citenamefont {Markham}, \citenamefont {Kashefi}, \citenamefont {Diamanti},
  \citenamefont {Coudreau}, \citenamefont {Milman}, \citenamefont {van Loock},\
  and\ \citenamefont {Ferrini}}]{DouceCVIQP2017}%
  \BibitemOpen
  \bibfield  {author} {\bibinfo {author} {\bibfnamefont {T.}~\bibnamefont
  {Douce}}, \bibinfo {author} {\bibfnamefont {D.}~\bibnamefont {Markham}},
  \bibinfo {author} {\bibfnamefont {E.}~\bibnamefont {Kashefi}}, \bibinfo
  {author} {\bibfnamefont {E.}~\bibnamefont {Diamanti}}, \bibinfo {author}
  {\bibfnamefont {T.}~\bibnamefont {Coudreau}}, \bibinfo {author}
  {\bibfnamefont {P.}~\bibnamefont {Milman}}, \bibinfo {author} {\bibfnamefont
  {P.}~\bibnamefont {van Loock}},\ and\ \bibinfo {author} {\bibfnamefont
  {G.}~\bibnamefont {Ferrini}},\ }\bibfield  {title} {\bibinfo {title}
  {Continuous-variable instantaneous quantum computing is hard to sample},\
  }\href {https://doi.org/10.1103/PhysRevLett.118.070503} {\bibfield  {journal}
  {\bibinfo  {journal} {Physical Review Letters}\ }\textbf {\bibinfo {volume}
  {118}},\ \bibinfo {pages} {070503} (\bibinfo {year} {2017})}\BibitemShut
  {NoStop}%
\bibitem [{\citenamefont {Chabaud}\ \emph {et~al.}(2017)\citenamefont
  {Chabaud}, \citenamefont {Douce}, \citenamefont {Markham}, \citenamefont {van
  Loock}, \citenamefont {Kashefi},\ and\ \citenamefont
  {Ferrini}}]{Chabaud2017hom}%
  \BibitemOpen
  \bibfield  {author} {\bibinfo {author} {\bibfnamefont {U.}~\bibnamefont
  {Chabaud}}, \bibinfo {author} {\bibfnamefont {T.}~\bibnamefont {Douce}},
  \bibinfo {author} {\bibfnamefont {D.}~\bibnamefont {Markham}}, \bibinfo
  {author} {\bibfnamefont {P.}~\bibnamefont {van Loock}}, \bibinfo {author}
  {\bibfnamefont {E.}~\bibnamefont {Kashefi}},\ and\ \bibinfo {author}
  {\bibfnamefont {G.}~\bibnamefont {Ferrini}},\ }\bibfield  {title} {\bibinfo
  {title} {Continuous-variable sampling from photon-added or photon-subtracted
  squeezed states},\ }\href {https://doi.org/10.1103/PhysRevA.96.062307}
  {\bibfield  {journal} {\bibinfo  {journal} {Physical Review A}\ }\textbf
  {\bibinfo {volume} {96}},\ \bibinfo {pages} {062307} (\bibinfo {year}
  {2017})}\BibitemShut {NoStop}%
\bibitem [{\citenamefont {Chakhmakhchyan}\ and\ \citenamefont
  {Cerf}(2017)}]{BShomodyne2017}%
  \BibitemOpen
  \bibfield  {author} {\bibinfo {author} {\bibfnamefont {L.}~\bibnamefont
  {Chakhmakhchyan}}\ and\ \bibinfo {author} {\bibfnamefont {N.~J.}\
  \bibnamefont {Cerf}},\ }\bibfield  {title} {\bibinfo {title} {Boson sampling
  with gaussian measurements},\ }\href
  {https://doi.org/10.1103/PhysRevA.96.032326} {\bibfield  {journal} {\bibinfo
  {journal} {Physical Review A}\ }\textbf {\bibinfo {volume} {96}},\ \bibinfo
  {pages} {032326} (\bibinfo {year} {2017})}\BibitemShut {NoStop}%
\bibitem [{\citenamefont {Gottesman}\ \emph
  {et~al.}(2001{\natexlab{b}})\citenamefont {Gottesman}, \citenamefont
  {Kitaev},\ and\ \citenamefont {Preskill}}]{GKP2001}%
  \BibitemOpen
  \bibfield  {author} {\bibinfo {author} {\bibfnamefont {D.}~\bibnamefont
  {Gottesman}}, \bibinfo {author} {\bibfnamefont {A.}~\bibnamefont {Kitaev}},\
  and\ \bibinfo {author} {\bibfnamefont {J.}~\bibnamefont {Preskill}},\
  }\bibfield  {title} {\bibinfo {title} {Encoding a qubit in an oscillator},\
  }\href {https://doi.org/10.1103/PhysRevA.64.012310} {\bibfield  {journal}
  {\bibinfo  {journal} {Physical Review A}\ }\textbf {\bibinfo {volume} {64}},\
  \bibinfo {pages} {012310} (\bibinfo {year} {2001}{\natexlab{b}})}\BibitemShut
  {NoStop}%
\bibitem [{\citenamefont {Banaszek}\ and\ \citenamefont
  {W\'odkiewicz}(1998)}]{Banaszek1998}%
  \BibitemOpen
  \bibfield  {author} {\bibinfo {author} {\bibfnamefont {K.}~\bibnamefont
  {Banaszek}}\ and\ \bibinfo {author} {\bibfnamefont {K.}~\bibnamefont
  {W\'odkiewicz}},\ }\bibfield  {title} {\bibinfo {title} {Nonlocality of the
  einstein-podolsky-rosen state in the wigner representation},\ }\href
  {https://doi.org/10.1103/PhysRevA.58.4345} {\bibfield  {journal} {\bibinfo
  {journal} {Physical Review A}\ }\textbf {\bibinfo {volume} {58}},\ \bibinfo
  {pages} {4345} (\bibinfo {year} {1998})}\BibitemShut {NoStop}%
\bibitem [{\citenamefont {Barbosa}(2015)}]{barbosa2015contextuality}%
  \BibitemOpen
  \bibfield  {author} {\bibinfo {author} {\bibfnamefont {R.~S.}\ \bibnamefont
  {Barbosa}},\ }\emph {\bibinfo {title} {Contextuality in quantum mechanics and
  beyond}},\ \href@noop {} {Ph.D. thesis},\ \bibinfo  {school} {University of
  Oxford} (\bibinfo {year} {2015})\BibitemShut {NoStop}%
\bibitem [{\citenamefont {Kenfack}\ and\ \citenamefont
  {{\.Z}yczkowski}(2004)}]{kenfack2004negativity}%
  \BibitemOpen
  \bibfield  {author} {\bibinfo {author} {\bibfnamefont {A.}~\bibnamefont
  {Kenfack}}\ and\ \bibinfo {author} {\bibfnamefont {K.}~\bibnamefont
  {{\.Z}yczkowski}},\ }\bibfield  {title} {\bibinfo {title} {Negativity of the
  {W}igner function as an indicator of non-classicality},\ }\href
  {https://doi.org/10.1088/1464-4266/6/10/003} {\bibfield  {journal} {\bibinfo
  {journal} {Journal of Optics B: Quantum and Semiclassical Optics}\ }\textbf
  {\bibinfo {volume} {6}},\ \bibinfo {pages} {396} (\bibinfo {year}
  {2004})}\BibitemShut {NoStop}%
\bibitem [{\citenamefont {Mari}\ \emph {et~al.}(2011)\citenamefont {Mari},
  \citenamefont {Kieling}, \citenamefont {Nielsen}, \citenamefont {Polzik},\
  and\ \citenamefont {Eisert}}]{mari2011directly}%
  \BibitemOpen
  \bibfield  {author} {\bibinfo {author} {\bibfnamefont {A.}~\bibnamefont
  {Mari}}, \bibinfo {author} {\bibfnamefont {K.}~\bibnamefont {Kieling}},
  \bibinfo {author} {\bibfnamefont {B.~M.}\ \bibnamefont {Nielsen}}, \bibinfo
  {author} {\bibfnamefont {E.}~\bibnamefont {Polzik}},\ and\ \bibinfo {author}
  {\bibfnamefont {J.}~\bibnamefont {Eisert}},\ }\bibfield  {title} {\bibinfo
  {title} {Directly estimating nonclassicality},\ }\href
  {https://doi.org/10.1103/physrevlett.106.010403} {\bibfield  {journal}
  {\bibinfo  {journal} {Physical Review Letters}\ }\textbf {\bibinfo {volume}
  {106}},\ \bibinfo {pages} {010403} (\bibinfo {year} {2011})}\BibitemShut
  {NoStop}%
\bibitem [{\citenamefont {Chabaud}\ \emph {et~al.}(2021)\citenamefont
  {Chabaud}, \citenamefont {Emeriau},\ and\ \citenamefont
  {Grosshans}}]{chabaud2021witnessing}%
  \BibitemOpen
  \bibfield  {author} {\bibinfo {author} {\bibfnamefont {U.}~\bibnamefont
  {Chabaud}}, \bibinfo {author} {\bibfnamefont {P.-E.}\ \bibnamefont
  {Emeriau}},\ and\ \bibinfo {author} {\bibfnamefont {F.}~\bibnamefont
  {Grosshans}},\ }\bibfield  {title} {\bibinfo {title} {Witnessing {W}igner
  negativity},\ }\href {https://doi.org/10.22331/q-2021-06-08-471} {\bibfield
  {journal} {\bibinfo  {journal} {{Quantum}}\ }\textbf {\bibinfo {volume}
  {5}},\ \bibinfo {pages} {471} (\bibinfo {year} {2021})}\BibitemShut {NoStop}%
\bibitem [{\citenamefont {Narcowich}\ and\ \citenamefont
  {O'Connell}(1986)}]{narcowich_necessary_1986}%
  \BibitemOpen
  \bibfield  {author} {\bibinfo {author} {\bibfnamefont {F.~J.}\ \bibnamefont
  {Narcowich}}\ and\ \bibinfo {author} {\bibfnamefont {R.~F.}\ \bibnamefont
  {O'Connell}},\ }\bibfield  {title} {\bibinfo {title} {Necessary and
  sufficient conditions for a phase-space function to be a {{Wigner}}
  distribution},\ }\href@noop {} {\bibfield  {journal} {\bibinfo  {journal}
  {Physical Review A}\ }\bibinfo {series} {3},\ \textbf {\bibinfo {volume}
  {34}} (\bibinfo {year} {1986})}\BibitemShut {NoStop}%
\bibitem [{\citenamefont {Narcowich}\ and\ \citenamefont
  {O'Connell}(1988)}]{narcowich_unified_1988}%
  \BibitemOpen
  \bibfield  {author} {\bibinfo {author} {\bibfnamefont {F.~J.}\ \bibnamefont
  {Narcowich}}\ and\ \bibinfo {author} {\bibfnamefont {R.~F.}\ \bibnamefont
  {O'Connell}},\ }\bibfield  {title} {\bibinfo {title} {A unified approach to
  quantum dynamical maps and gaussian {{Wigner}} distributions},\ }\href
  {https://doi.org/10.1016/0375-9601(88)91009-2} {\bibfield  {journal}
  {\bibinfo  {journal} {Physics Letters A}\ }\textbf {\bibinfo {volume}
  {133}},\ \bibinfo {pages} {167} (\bibinfo {year} {1988})}\BibitemShut
  {NoStop}%
\bibitem [{\citenamefont {Br{\"o}cker}\ and\ \citenamefont
  {Werner}(1995)}]{brocker_mixed_1995}%
  \BibitemOpen
  \bibfield  {author} {\bibinfo {author} {\bibfnamefont {T.}~\bibnamefont
  {Br{\"o}cker}}\ and\ \bibinfo {author} {\bibfnamefont {R.~F.}\ \bibnamefont
  {Werner}},\ }\bibfield  {title} {\bibinfo {title} {Mixed states with positive
  {{Wigner}} functions},\ }\href {https://doi.org/10.1063/1.531326} {\bibfield
  {journal} {\bibinfo  {journal} {Journal of Mathematical Physics}\ }\textbf
  {\bibinfo {volume} {36}},\ \bibinfo {pages} {62} (\bibinfo {year}
  {1995})}\BibitemShut {NoStop}%
\bibitem [{\citenamefont {Abramsky}\ \emph {et~al.}(2012)\citenamefont
  {Abramsky}, \citenamefont {Mansfield},\ and\ \citenamefont
  {Barbosa}}]{abramsky2012cohomology}%
  \BibitemOpen
  \bibfield  {author} {\bibinfo {author} {\bibfnamefont {S.}~\bibnamefont
  {Abramsky}}, \bibinfo {author} {\bibfnamefont {S.}~\bibnamefont
  {Mansfield}},\ and\ \bibinfo {author} {\bibfnamefont {R.~S.}\ \bibnamefont
  {Barbosa}},\ }\bibfield  {title} {\bibinfo {title} {The cohomology of
  non-locality and contextuality},\ }in\ \href
  {https://doi.org/10.4204/EPTCS.95.1} {\emph {\bibinfo {booktitle}
  {Proceedings of 8th International Workshop on Quantum Physics and Logic (QPL
  2011)}}},\ \bibinfo {series} {Electronic Proceedings in Theoretical Computer
  Science}, Vol.~\bibinfo {volume} {95},\ \bibinfo {editor} {edited by\
  \bibinfo {editor} {\bibfnamefont {B.}~\bibnamefont {Jacobs}}, \bibinfo
  {editor} {\bibfnamefont {P.}~\bibnamefont {Selinger}},\ and\ \bibinfo
  {editor} {\bibfnamefont {B.}~\bibnamefont {Spitters}}}\ (\bibinfo
  {publisher} {Open Publishing Association},\ \bibinfo {year} {2012})\ pp.\
  \bibinfo {pages} {1--14}\BibitemShut {NoStop}%
\bibitem [{\citenamefont {Abramsky}\ \emph {et~al.}(2015)\citenamefont
  {Abramsky}, \citenamefont {Barbosa}, \citenamefont {Kishida}, \citenamefont
  {Lal},\ and\ \citenamefont {Mansfield}}]{abramsky2015contextuality}%
  \BibitemOpen
  \bibfield  {author} {\bibinfo {author} {\bibfnamefont {S.}~\bibnamefont
  {Abramsky}}, \bibinfo {author} {\bibfnamefont {R.~S.}\ \bibnamefont
  {Barbosa}}, \bibinfo {author} {\bibfnamefont {K.}~\bibnamefont {Kishida}},
  \bibinfo {author} {\bibfnamefont {R.}~\bibnamefont {Lal}},\ and\ \bibinfo
  {author} {\bibfnamefont {S.}~\bibnamefont {Mansfield}},\ }\bibfield  {title}
  {\bibinfo {title} {Contextuality, cohomology and paradox},\ }in\ \href
  {https://doi.org/10.4230/LIPIcs.CSL.2015.211} {\emph {\bibinfo {booktitle}
  {24th {EACSL} {Annual} {Conference} on {Computer} {Science} {Logic} ({CSL}
  2015)}}},\ \bibinfo {series} {Leibniz {International} {Proceedings} in
  {Informatics} ({LIPIcs})}, Vol.~\bibinfo {volume} {41},\ \bibinfo {editor}
  {edited by\ \bibinfo {editor} {\bibfnamefont {S.}~\bibnamefont {Kreutzer}}}\
  (\bibinfo  {publisher} {Schloss Dagstuhl–Leibniz-Zentrum fuer Informatik},\
  \bibinfo {year} {2015})\ pp.\ \bibinfo {pages} {211--228}\BibitemShut
  {NoStop}%
\bibitem [{\citenamefont {Car{\`u}}(2017)}]{caru2017}%
  \BibitemOpen
  \bibfield  {author} {\bibinfo {author} {\bibfnamefont {G.}~\bibnamefont
  {Car{\`u}}},\ }\bibfield  {title} {\bibinfo {title} {On the cohomology of
  contextuality},\ }in\ \href {https://doi.org/10.4204/EPTCS.236.2} {\emph
  {\bibinfo {booktitle} {13th International Conference on Quantum Physics and
  Logic (QPL 2016)}}},\ \bibinfo {series} {Electronic Proceedings in
  Theoretical Computer Science}, Vol.\ \bibinfo {volume} {236},\ \bibinfo
  {editor} {edited by\ \bibinfo {editor} {\bibfnamefont {R.}~\bibnamefont
  {Duncan}}\ and\ \bibinfo {editor} {\bibfnamefont {C.}~\bibnamefont
  {Heunen}}}\ (\bibinfo  {publisher} {Open Publishing Association},\ \bibinfo
  {year} {2017})\ pp.\ \bibinfo {pages} {21--39}\BibitemShut {NoStop}%
\bibitem [{\citenamefont {Car{\`u}}(2018)}]{caru2018towards}%
  \BibitemOpen
  \bibfield  {author} {\bibinfo {author} {\bibfnamefont {G.}~\bibnamefont
  {Car{\`u}}},\ }\href@noop {} {\bibinfo {title} {Towards a complete cohomology
  invariant for non-locality and contextuality}},\ \bibinfo {howpublished}
  {Preprint \href{https://arxiv.org/abs/1807.04203}{arXiv:1807.04203
  [quant-ph]}} (\bibinfo {year} {2018})\BibitemShut {NoStop}%
\bibitem [{\citenamefont {Raussendorf}(2019)}]{raussendorf_cohomological_2019}%
  \BibitemOpen
  \bibfield  {author} {\bibinfo {author} {\bibfnamefont {R.}~\bibnamefont
  {Raussendorf}},\ }\bibfield  {title} {\bibinfo {title} {Cohomological
  framework for contextual quantum computations},\ }\href
  {https://doi.org/10.26421/QIC19.13-14-4} {\bibfield  {journal} {\bibinfo
  {journal} {Quantum Information and Computation}\ }\textbf {\bibinfo {volume}
  {19}},\ \bibinfo {pages} {1141} (\bibinfo {year} {2019})}\BibitemShut
  {NoStop}%
\bibitem [{\citenamefont {Okay}\ and\ \citenamefont
  {Sheinbaum}(2021)}]{okay_classifying_2021}%
  \BibitemOpen
  \bibfield  {author} {\bibinfo {author} {\bibfnamefont {C.}~\bibnamefont
  {Okay}}\ and\ \bibinfo {author} {\bibfnamefont {D.}~\bibnamefont
  {Sheinbaum}},\ }\bibfield  {title} {\bibinfo {title} {Classifying space for
  quantum contextuality},\ }\href {https://doi.org/10.1007/s00023-020-00993-3}
  {\bibfield  {journal} {\bibinfo  {journal} {Annales Henri Poincar\'e}\
  }\textbf {\bibinfo {volume} {22}},\ \bibinfo {pages} {529} (\bibinfo {year}
  {2021})},\ \Eprint {https://arxiv.org/abs/1905.07723} {1905.07723}
  \BibitemShut {NoStop}%
\bibitem [{\citenamefont {Plastino}\ and\ \citenamefont
  {Cabello}(2010)}]{plastino2010state}%
  \BibitemOpen
  \bibfield  {author} {\bibinfo {author} {\bibfnamefont {{\'A}.~R.}\
  \bibnamefont {Plastino}}\ and\ \bibinfo {author} {\bibfnamefont
  {A.}~\bibnamefont {Cabello}},\ }\bibfield  {title} {\bibinfo {title}
  {State-independent quantum contextuality for continuous variables},\ }\href
  {https://doi.org/10.1103/PhysRevA.82.022114} {\bibfield  {journal} {\bibinfo
  {journal} {Physical Review A}\ }\textbf {\bibinfo {volume} {82}},\ \bibinfo
  {pages} {022114} (\bibinfo {year} {2010})}\BibitemShut {NoStop}%
\bibitem [{\citenamefont {He}\ \emph {et~al.}(2010)\citenamefont {He},
  \citenamefont {Cavalcanti}, \citenamefont {Reid},\ and\ \citenamefont
  {Drummond}}]{he2010bell}%
  \BibitemOpen
  \bibfield  {author} {\bibinfo {author} {\bibfnamefont {Q.-Y.}\ \bibnamefont
  {He}}, \bibinfo {author} {\bibfnamefont {E.~G.}\ \bibnamefont {Cavalcanti}},
  \bibinfo {author} {\bibfnamefont {M.~D.}\ \bibnamefont {Reid}},\ and\
  \bibinfo {author} {\bibfnamefont {P.~D.}\ \bibnamefont {Drummond}},\
  }\bibfield  {title} {\bibinfo {title} {Bell inequalities for
  continuous-variable measurements},\ }\href
  {https://doi.org/10.1103/PhysRevA.81.062106} {\bibfield  {journal} {\bibinfo
  {journal} {Physical Review A}\ }\textbf {\bibinfo {volume} {81}},\ \bibinfo
  {pages} {062106} (\bibinfo {year} {2010})}\BibitemShut {NoStop}%
\bibitem [{\citenamefont {McKeown}\ \emph {et~al.}(2011)\citenamefont
  {McKeown}, \citenamefont {Paris},\ and\ \citenamefont
  {Paternostro}}]{mckeown2011testing}%
  \BibitemOpen
  \bibfield  {author} {\bibinfo {author} {\bibfnamefont {G.}~\bibnamefont
  {McKeown}}, \bibinfo {author} {\bibfnamefont {M.~G.}\ \bibnamefont {Paris}},\
  and\ \bibinfo {author} {\bibfnamefont {M.}~\bibnamefont {Paternostro}},\
  }\bibfield  {title} {\bibinfo {title} {Testing quantum contextuality of
  continuous-variable states},\ }\href
  {https://doi.org/10.1103/PhysRevA.83.062105} {\bibfield  {journal} {\bibinfo
  {journal} {Physical Review A}\ }\textbf {\bibinfo {volume} {83}},\ \bibinfo
  {pages} {062105} (\bibinfo {year} {2011})}\BibitemShut {NoStop}%
\bibitem [{\citenamefont {Asadian}\ \emph {et~al.}(2015)\citenamefont
  {Asadian}, \citenamefont {Budroni}, \citenamefont {Steinhoff}, \citenamefont
  {Rabl},\ and\ \citenamefont {G{\"u}hne}}]{asadian2015contextuality}%
  \BibitemOpen
  \bibfield  {author} {\bibinfo {author} {\bibfnamefont {A.}~\bibnamefont
  {Asadian}}, \bibinfo {author} {\bibfnamefont {C.}~\bibnamefont {Budroni}},
  \bibinfo {author} {\bibfnamefont {F.~E.}\ \bibnamefont {Steinhoff}}, \bibinfo
  {author} {\bibfnamefont {P.}~\bibnamefont {Rabl}},\ and\ \bibinfo {author}
  {\bibfnamefont {O.}~\bibnamefont {G{\"u}hne}},\ }\bibfield  {title} {\bibinfo
  {title} {Contextuality in phase space},\ }\href
  {https://doi.org/10.1103/PhysRevLett.114.250403} {\bibfield  {journal}
  {\bibinfo  {journal} {Physical Review Letters}\ }\textbf {\bibinfo {volume}
  {114}},\ \bibinfo {pages} {250403} (\bibinfo {year} {2015})}\BibitemShut
  {NoStop}%
\bibitem [{\citenamefont {Laversanne-Finot}\ \emph {et~al.}(2017)\citenamefont
  {Laversanne-Finot}, \citenamefont {Ketterer}, \citenamefont {Barros},
  \citenamefont {Walborn}, \citenamefont {Coudreau}, \citenamefont {Keller},\
  and\ \citenamefont {Milman}}]{laversanne2017general}%
  \BibitemOpen
  \bibfield  {author} {\bibinfo {author} {\bibfnamefont {A.}~\bibnamefont
  {Laversanne-Finot}}, \bibinfo {author} {\bibfnamefont {A.}~\bibnamefont
  {Ketterer}}, \bibinfo {author} {\bibfnamefont {M.~R.}\ \bibnamefont
  {Barros}}, \bibinfo {author} {\bibfnamefont {S.~P.}\ \bibnamefont {Walborn}},
  \bibinfo {author} {\bibfnamefont {T.}~\bibnamefont {Coudreau}}, \bibinfo
  {author} {\bibfnamefont {A.}~\bibnamefont {Keller}},\ and\ \bibinfo {author}
  {\bibfnamefont {P.}~\bibnamefont {Milman}},\ }\bibfield  {title} {\bibinfo
  {title} {General conditions for maximal violation of non-contextuality in
  discrete and continuous variables},\ }\href
  {https://doi.org/10.1088/1751-8121/aa6016} {\bibfield  {journal} {\bibinfo
  {journal} {Journal of Physics A: Mathematical and Theoretical}\ }\textbf
  {\bibinfo {volume} {50}},\ \bibinfo {pages} {155304} (\bibinfo {year}
  {2017})}\BibitemShut {NoStop}%
\bibitem [{\citenamefont {Ketterer}\ \emph {et~al.}(2018)\citenamefont
  {Ketterer}, \citenamefont {Laversanne-Finot},\ and\ \citenamefont
  {Aolita}}]{ketterer2018continuous}%
  \BibitemOpen
  \bibfield  {author} {\bibinfo {author} {\bibfnamefont {A.}~\bibnamefont
  {Ketterer}}, \bibinfo {author} {\bibfnamefont {A.}~\bibnamefont
  {Laversanne-Finot}},\ and\ \bibinfo {author} {\bibfnamefont {L.}~\bibnamefont
  {Aolita}},\ }\bibfield  {title} {\bibinfo {title} {Continuous-variable
  supraquantum nonlocality},\ }\href
  {https://doi.org/10.1103/PhysRevA.97.012133} {\bibfield  {journal} {\bibinfo
  {journal} {Physical Review A}\ }\textbf {\bibinfo {volume} {97}},\ \bibinfo
  {pages} {012133} (\bibinfo {year} {2018})}\BibitemShut {NoStop}%
\bibitem [{\citenamefont {Wenger}\ \emph {et~al.}(2003)\citenamefont {Wenger},
  \citenamefont {Hafezi}, \citenamefont {Grosshans}, \citenamefont
  {Tualle-Brouri},\ and\ \citenamefont {Grangier}}]{Fred2003}%
  \BibitemOpen
  \bibfield  {author} {\bibinfo {author} {\bibfnamefont {J.}~\bibnamefont
  {Wenger}}, \bibinfo {author} {\bibfnamefont {M.}~\bibnamefont {Hafezi}},
  \bibinfo {author} {\bibfnamefont {F.}~\bibnamefont {Grosshans}}, \bibinfo
  {author} {\bibfnamefont {R.}~\bibnamefont {Tualle-Brouri}},\ and\ \bibinfo
  {author} {\bibfnamefont {P.}~\bibnamefont {Grangier}},\ }\bibfield  {title}
  {\bibinfo {title} {Maximal violation of bell inequalities using
  continuous-variable measurements},\ }\href
  {https://doi.org/10.1103/PhysRevA.67.012105} {\bibfield  {journal} {\bibinfo
  {journal} {Physical Review A}\ }\textbf {\bibinfo {volume} {67}},\ \bibinfo
  {pages} {012105} (\bibinfo {year} {2003})}\BibitemShut {NoStop}%
\bibitem [{\citenamefont {Haferkamp}\ and\ \citenamefont
  {Bermejo-Vega}(2021)}]{haferkamp2021equivalence}%
  \BibitemOpen
  \bibfield  {author} {\bibinfo {author} {\bibfnamefont {J.}~\bibnamefont
  {Haferkamp}}\ and\ \bibinfo {author} {\bibfnamefont {J.}~\bibnamefont
  {Bermejo-Vega}},\ }\href@noop {} {\bibinfo {title} {Equivalence of
  contextuality and wigner function negativity in continuous-variable quantum
  optics}},\ \bibinfo {howpublished} {Preprint
  \href{https://arxiv.org/abs/2112.14788}{arXiv:2112.14788 [quant-ph]}}
  (\bibinfo {year} {2021})\BibitemShut {NoStop}%
\bibitem [{\citenamefont {Tychonoff}(1930)}]{tychonoff1930topologische}%
  \BibitemOpen
  \bibfield  {author} {\bibinfo {author} {\bibfnamefont {A.}~\bibnamefont
  {Tychonoff}},\ }\bibfield  {title} {\bibinfo {title} {{\"U}ber die
  topologische erweiterung von r{\"a}umen},\ }\href
  {https://doi.org/10.1007/BF01782364} {\bibfield  {journal} {\bibinfo
  {journal} {Mathematische Annalen}\ }\textbf {\bibinfo {volume} {102}},\
  \bibinfo {pages} {544} (\bibinfo {year} {1930})}\BibitemShut {NoStop}%
\bibitem [{\citenamefont {Hall}(2013)}]{hall2013quantum}%
  \BibitemOpen
  \bibfield  {author} {\bibinfo {author} {\bibfnamefont {B.~C.}\ \bibnamefont
  {Hall}},\ }\href {https://doi.org/10.1007/978-1-4614-7116-5} {\emph {\bibinfo
  {title} {Quantum theory for mathematicians}}},\ Vol.\ \bibinfo {volume}
  {267}\ (\bibinfo  {publisher} {Springer},\ \bibinfo {year}
  {2013})\BibitemShut {NoStop}%
\bibitem [{\citenamefont {Fournier}\ and\ \citenamefont
  {Printems}(2010)}]{fournier_absolute_2010}%
  \BibitemOpen
  \bibfield  {author} {\bibinfo {author} {\bibfnamefont {N.}~\bibnamefont
  {Fournier}}\ and\ \bibinfo {author} {\bibfnamefont {J.}~\bibnamefont
  {Printems}},\ }\bibfield  {title} {\bibinfo {title} {Absolute continuity for
  some one-dimensional processes},\ }\bibfield  {journal} {\bibinfo  {journal}
  {Bernoulli}\ }\textbf {\bibinfo {volume} {16}},\ \href
  {https://doi.org/10.3150/09-BEJ215} {10.3150/09-BEJ215} (\bibinfo {year}
  {2010}),\ \Eprint {https://arxiv.org/abs/0804.3037} {arXiv:0804.3037}
  \BibitemShut {NoStop}%
\end{thebibliography}
%

\pagebreak 
\appendix

\section*{Appendices}

 \noindent Here we give some background and discussion for the results stated in the main text. It is structured as follows:
 \begin{itemize}
 \item tools of measure theory;
 \item Wigner functions and the symplectic phase space;
 \item continuous-variable contextuality;
 \item "measuring" a displacement operator.
 \end{itemize}

\section{Tools of measure theory}
\label{app:measure}

\noindent In this appendix, we briefly introduce the necessary tools of measure theory. We refer the reader to~\cite{billingsley2008probability,tao2011introduction} for a more in-depth treatment. 

\medskip

To avoid some pathological behaviours when dealing with probability distributions on a continuum, we first need to define $\sigma$-algebras which give rise to a good notion of measurability. 
\begin{definition}[$\sigma$-algebras]
A $\sigma$-algebra on a set $U$ is a family $\Bc$ of subsets of U containing the empty set and closed under complementation and countable unions, that is:
\begin{enumerate}[label=(\roman*)]
    \item $\emptyset \in \Bc$.
    \item for all $E \in \Bc$, $E^c \in \Bc$.
    \item for all $E_1,E_2, \dots \in \Bc$, $\cup_{i=1}^\infty E_i \in \Bc$.
\end{enumerate}
\end{definition}

\begin{definition}[Measurable space]
A measurable space is a pair $\bm U = \tuple{U,\Fc_U}$ consisting of a set $U$ and a $\sigma$-algebra $\Fc_U$ on $U$.
\end{definition}
\noindent In some sense, the $\sigma$-algebra specifies the subsets of $U$ that can be assigned a `size', and which are therefore called the \emph{measurable sets} of $\bm U$. We use boldface to refer to measurable spaces and regular font to refer to the underlying set. 
A trivial example of a $\sigma$-algebra over any set $U$ is its powerset $\Pc(U)$, which gives the discrete measurable space $\tuple{U,\Pc(U)}$, where every set is measurable.
This is typically used when $U$ is countable (finite or countably infinite), in which case this discrete $\sigma$-algebra is generated by the singletons.
Another example, of central importance in measure theory, is the Borel $\sigma$-algebra $\Bc_\R$ generated from the open sets of $\R$, whose elements are called the Borel sets. 
This gives the measurable space $\tuple{\R,\Bc_\R}$.
Working with Borel sets avoids the problems that would arise if we naively attempted to measure or assign probabilities to points in the continuum.

We also need to deal with measurable spaces formed by taking the product of an uncountably infinite family of measurable spaces. 
As enlightened by Tychonoff's theorem \cite{tychonoff1930topologische}, we use the \textit{product $\sigma$-algebra}. 
\begin{definition}[Infinite product]\index{Measurable!space!infinite product of|textbf}
Fix a possibly uncountably infinite index set $I$.
The product of measurable spaces $(\bm U_i = \tuple{U_i,\Fc_i})_{i \in I}$ is the measurable space:
\[\prod_{i \in I} \bm U_i = \tuple{\prod_{i \in I} U_i, \prod_{i \in I} \Fc_i} = \tuple{U_I,\Fc_I} \Mcomma\]
where $U_I = \prod_{i \in I} U_i$ is the Cartesian product of the underlying sets, and the $\sigma$-algebra $\Fc_I = \prod_{i \in I} \Fc_i$ is obtained via the product construction \ie it is generated by subsets of $\prod_{i \in I} U_i$ of the form $\prod_{i \in I} B_i$ where for all $i \in I$, $B_i \subseteq U_i$ and $B_i \subsetneq U_i$ only for a finite number of $i \in I$.
\label{def:productmeasurableinfinite}
\end{definition}
Remarkably, the product topology is the smallest topology that makes the projection maps $\fdec{\pi_k}{\prod_{i \in I} U_i }{U_k}$ measurable.
This definition reduces straightforwardly to the case of a finite product.
We can now formally define measurable functions and measures on those spaces. 
\begin{definition}[Measurable function]\index{Measurable!function|textbf}
A measurable function between measurable spaces $\bm U = \tuple{U,\Fc_U}$ and $\bm V = \tuple{V, \Fc_V}$ is a function $\fdec{f}{U}{V}$ whose preimage preserves measurable sets, \ie such that, for any $E \in \Fc_V$, ${f^{-1}(E) \in \Fc_U}$.
\end{definition}
\noindent This is similar to the definition of a continuous function between topological spaces.
Measurable functions compose as expected.

\begin{definition}[Measure]
\label{def:measure}
A measure on a measurable space $\bm U = \tuple{U,\Fc_U}$ is a function $\fdec{\mu}{\Fc_U}{\Rext}$ from the $\sigma$-algebra to the extended real numbers $\Rext = \R \cup \enset{-\infty,+\infty}$ satisfying:
\begin{enumerate}[label=(\roman*)]
\item\label{enum:nonnegativity} (nonnegativity)
$\mu(E)\geq 0$ for all $E\in\Fc_U$;
\item (null empty set)
$\mu(\emptyset)=0$;
\item ($\sigma$-additivity)
for any countable family $\family{E_i}_{i=1}^\infty$ of pairwise disjoint measurable sets, $\mu(\bigcup_{i=1}^\infty E_i) = \sum_{i=1}^\infty \mu(E_i)$.
\end{enumerate}
\end{definition}
\noindent In particular, a measure $\mu$ on $\bm U = \tuple{U,\Fc_U}$ allows one to integrate well-behaved measurable functions $\fdec{f}{U}{\R}$ (where $\R$ is equipped with its Borel $\sigma$-algebra $\Bc_\R$) to obtain a real value, denoted
\begin{equation}
    \intg{\bm U}{f}{\mu} \; \text{ or } \; \intg{x \in U}{f(x)}{\mu(x)}.
\end{equation}
A simple example of a measurable function is the \emph{indicator function} of a measurable set $E \in \Fc_U$:
\[ \bm 1{_E}(x) \defeq \begin{cases} 1 & \text{if $x \in E$} \\ 0 & \text{if $x \not\in E$.}\end{cases}\]
For any measure $\mu$ on $\bm U$, its integral yields the `size' of E:
\begin{equation}\label{eq:integralindicator}
\intg{\bm U}{\bm 1_{E}}{\mu} = \mu(E).
\end{equation}
A measure $\mu$ on a measurable space $\bm U$ is said to be \emph{finite} if $\mu(U)<\infty$ and it is a \emph{probability measure} if it is of unit mass \ie $\mu(U)=1$.

A measurable function $f$ between measurable spaces $\bm U$ and $ \bm V$ carries any measure $\mu$ on $\bm U$ to a measure $f_*\mu$ on $\bm V$. 
This is known as a \emph{push-forward} operation.
This push-forward measure is given by $f_*\mu(E) = \mu(f^{-1}(E))$ for any set $E$ measurable in $\bm V$.
An important use of push-forward measures is that for any integrable function $g$ between measurable spaces $\bm V$ and $\tuple{\R,\Bc_\R}$, one can write the following change-of-variables formula:
\begin{equation}\label{eq:changeofvariables}
\intg{\bm U}{g \circ f}{\mu} = \intg{\bm V}{g}{f_*\mu}.
\end{equation}
A case that is of particular interest to us is the push-forward of a measure $\mu$ on a product space $\bm U_1 \times \bm U_2$ along a projection  $\fdec{\pi_i}{ U_1 \times  U_2}{U_i}$. This yields the \emph{marginal measure} $\mu|_{\bm U_i}={\pi_i}_*\mu$, where for any $E \in \Fc_{U_1}$ measurable,  $\mu|_{\bm U_1}(E) = \mu(\pi_1^{-1}(E)) = \mu(E \times U_2)$.
In the opposite direction, given a measure $\mu_1$ on $\bm U_1$ and a measure $\mu_2$ on  $\bm U_2$, a \emph{product measure} $\mu_1 \times \mu_2$ is a measure on the product measurable space $\bm U_1 \times \bm U_2$ satisfying $(\mu_1 \times \mu_2)(E_1 \times E_2) = \mu_1(E_1)\mu_2(E_2)$ for all $E_1 \in \Fc_1$ and $E_2 \in \Fc_2$. For probability measures, there is a uniquely determined product measure.

The last ingredient we need from measure theory is called a \emph{Markov kernel}. 
\begin{definition}[Markov kernel]
A Markov kernel (or probability kernel) between measurable spaces $\bm U = \tuple{U,\Fc_U}$
and $\bm V = \tuple{V, \Fc_V}$ is a function
$\fdec{k}{U \times \Fc_V}{[0,1]}$ (the space $[0,1]$ is assumed to be equipped with its Borel $\sigma$-algebra)
such that:
\begin{enumerate}[label=(\roman*)]
    \item for all $E \in \Fc_V$, $\fdec{k(\dummy,E)}{U}{[0,1]}$ is a measurable function;
    \item for all $x \in U$, $\fdec{k(x,\dummy)}{\Fc_V}{[0,1]}$ is a probability measure.
\end{enumerate}
\end{definition}
\noindent Markov kernels generalise the discrete notion of stochastic matrices.

Finally, we need the following result, which allows us control over a probability measure via its Fourier transform:
\begin{lemma}
  \label{lem:density_L2}
  Let $M\ge1$, let \(\mu\) be a Borel probability measure on $\R^{2M}$ and define its
  Fourier transform by \(\operatorname{FT}[\mu](\bm{x}) = \int_{\R^{2M}}e^{i\bm{x}
    \cdot \bm{y}} \dd{\mu(\bm{y})}\) for \(\bm x \in\R^{2M}\). If
  \(\operatorname{FT}[\mu] \in L^2(\R^{2M})\) then \(\mu\) has a density function \(w_\mu \in
  L^1(\R^{2M}) \cap L^2(\R^{2M})\).
\end{lemma}

\begin{proof}
  We mostly follow the proof of \cite[Lem. 1.1]{fournier_absolute_2010}, though
  adapted to our setting and conventions.

  For \(n \in \N^*\), consider the convolution \(\mu_n \coloneqq \mu \star g_n\)
  with a centered Gaussian distribution \(g_n\), i.e.\ with density function
  \((\sfrac{n}{2\pi})^{d/2} \exp(-\sfrac{n\|\bm x\|^2}{2})\), and in particular
  \(|\text{FT}[\mu_n](\bm x)| = \exp(-\sfrac{\|\bm x\|^2}{2n}) |\text{FT}[\mu](\bm
  x)| \leqslant |\text{FT}[\mu](\bm x)|\) for all \(\bm x \in\R^{2M}\). For any \(n
  \in \N^*\), \(\mu_n\) has a density \(f_n \in L^1(\R^{2M}) \cap L^\infty(\R^{2M})\)),
  so we can apply the Plancherel theorem to obtain:
  \begin{equation}
    \int_{\R^{2M}} f_n(\bm x)^2 \,\dd{\bm x}
    = \frac{1}{(\sqrt{2\pi})^M} \int_{\R^{2M}} |\text{FT}[\mu_n](\bm x)|^2 \,\dd{\bm x}
    \leqslant \frac{1}{(\sqrt{2\pi})^M} \int_{\R^{2M}} |\text{FT}[\mu](\bm x)|^2 \,\dd{\bm x}
    = C < \infty.
  \end{equation}
  By the Banach--Alaoglu theorem, the (closed) \(C\)-ball of \(L^2(\R^{2M})\) is
  sequentially compact in the weak topology, so that we can extract a
  subsequence \(n_k\) and some \(f \in L^2(\R^{2M})\) for which \(f_{n_k} \to f\)
  weakly. In particular, this implies that for any \(g \subseteq L^2(\R^{2M})\),
  and in particular any continuous function with compact support,
  \begin{equation}
    \int_{\R^{2M}} g(\bm x) \,\dd{\mu_{n_k}(\bm x)}
    = \int_{\R^{2M}} g(\bm x) f_{n_k}(\bm x) \,\dd{\bm x}
    \longrightarrow \int_{\R^{2M}} g(\bm x) f(\bm x) \,\dd{\bm x},
  \end{equation}
  i.e.\ \(\mu_n\) converges vaguely to the Radon measure \(f \dd{\bm x}\).
  Furthermore, for any compact \(E \subseteq \R^{2M}\) and \(n \in \N^*\), we
  have \(\mu_n(E) \leqslant 1\). Then, we must have \(\int_E f \dd{\bm x}
  \leqslant 1\), and \(f\dd{x}\) must be a bounded measure, which in turn is
  possible only if \(f \in L^1(\R^{2M})\).

  On the other hand, for all \(\bm x \in \R^{2M}\), \(\text{FT}[\mu_n] \to
  \text{FT}[\mu]\), so that \(\mu_n\) converges strictly, thus vaguely, to
  \(\mu\). Then, by the uniqueness of vague limits of sequences of bounded
  measures, \(\mu = f \dd{\bm x}\).
\end{proof}

\section{Wigner functions and the symplectic phase space}
\label{sec:wigner}

\noindent In this section, we recall some definitions from the main text and introduce additional preliminary material relating to phase space and Wigner functions.

\medskip

We fix $M \in \N^*$ to be the number of qumodes, that is, $M$
CV quantum systems. For a single qumode, the corresponding state
space is the Hilbert space of square-integrable functions \(L^2(\R)\) and the total Hilbert space for all \(M\) qumodes is then
\(L^2(\R)^{\otimes M} \cong L^2(\R^{M})\).

To each qumode, we associate the
usual position and momentum operators. These are defined on the dense
subspace of Schwartz functions (functions whose derivatives decrease faster than
any polynomial at infinity) by:
\begin{equation}
    \hat{q} \psi(x) \defeq x \psi(x) \quad \text{and} \quad \hat{p} \psi(x) \defeq -i \frac {\partial \psi(x)}{\partial x} \, ,
\end{equation}
where we have set $\hbar=1$. We write $\hat q_k$ and $\hat p_k$ to denote the position and momentum operators of the $k^{th}$ qumode, and we extend the definition above by linearity to any linear combination thereof. Any \(\R\)-linear combination of these operators is called a \emph{quadrature}. 

\subsection*{Phase space} 

\noindent The Wigner representation of a quantum state in the Hilbert space \(L^2(\R^{M})\) is a function defined on the \emph{phase space} \(\R^{2M}\), which can be intuitively understood as a quantum version of the position and momentum phase space of a classical particle. 
We equip this phase space with a symplectic form denoted ${\Omega}$: for $\bm
x,\bm y \in \R^{2M}$,
\begin{equation}
    \label{eq:symplecticform}
    \Omega(\bm x,\bm y) := \bm x \cdot J \bm y \quad \text{where} \quad J = \begin{pmatrix}
    0 & \Id_M \\
    - \Id_M & 0 
    \end{pmatrix},
\end{equation}
in a given basis \((\bm{e}_k,\bm{f}_k)_{k=1}^{M}\) of \(\R^{2M}\), which is
therefore a symplectic basis for the phase space.
We have that $J^{-1} = J^T = -J$ and $J$ can be seen as a linear map from $\R^{2M}$ to  $\R^{2M}$.  We also equip $\R^{2M}$ with its usual scalar product denoted by $\dummy \cdot \dummy$.

A \emph{Lagrangian vector subspace} is defined as a maximal isotropic subspace, that is, a maximal subspace on which the symplectic form $\Omega$ vanishes. For a symplectic space of dimension $2M$, Lagrangian subspaces are of dimension $M$. See \cite{Sudarshan1988} for a concise introduction to the symplectic structure of phase space and \cite{Gosson2006symplectic} for a detailed review.

The importance of the symplectic phase space $\R^{2M}$ comes from its relation
to the position and momentum quadrature operators. To any \(\bm{x} \in \R^{2M}\) we
associate a quadrature operator as follows. Assume w.l.o.g.\ that \(\bm{x} =
\sum_{k}q_k \bm{e}_k + \sum_{k}p_k \bm{f}_k\), and put
\begin{equation}
  \label{eq:generalised_quadrature}
  \hat{\bm{x}} = \sum_{k=1}^M q_k \hat{q}_k + \sum_{k=1}^{M} p_k \hat{p}_k,
\end{equation}
where the indices indicate on which qumode each operator acts. Then, it is
straightforward to verify, using the canonical commutation relations, that
\begin{equation}
  \label{eq:commutation_relation}
    [\hat{\bm{x}},\hat{\bm{y}}] = i\Omega(\bm{x},\bm{y}) \hat{\Id},
\end{equation}
\ie the symplectic structure encodes the commutation relations of quadrature operators.

The elements of \(\R^{2M}\) can also be associated to translations in phase
space. 
Firstly, for any \(s \in \R^{M}\), define the \emph{Weyl operators}, acting on
\(L^2(\R^{M})\), by
\begin{equation}
    \hat X(s)\psi(t) = \psi(t-s) \quad \text{and} \quad \hat Z(s)\psi(t) = e^{ist} \psi(t),
\end{equation}
for all $t\in\mathbb R^M$.
Then, define the \emph{displacement operator} for any \(\bm x = (q,p) \in \R^M
\times \R^M\) in the symplectic basis \((\bm{e}_k,\bm{f}_k)_{k=1}^{M}\) by
\begin{equation}
    \hat D(\bm{x}) = e^{-i\frac{q \cdot p}{2}}\hat X(q)\hat Z(p),
\end{equation}
so that
\begin{equation}
    [\hat D(\bm{x}),\hat D(\bm{y})] = e^{i \Omega(\bm{x},\bm{y})} \hat{\Id}.
\end{equation}

A quadrature operator such as the position operator $\hat q$ is self-adjoint and
can be expanded via the spectral theorem \cite[Th. 8.10]{hall2013quantum} as:
\begin{equation}
    \hat q = \intg{\R}{x}{P_{\hat q}(x)} \, ,
\end{equation}
where $P_{\hat q}$ is the projection-valued measure (PVM) associated to $\hat
q$. For any Lebesgue-measurable set \(E \subseteq \R\), $P_{\hat q}(E)$ is given
by \cite[Def. 8.8]{hall2013quantum}:
\begin{equation}
    P_{\hat q}(E) \psi(x) =
    \begin{cases}
      \psi(x) \quad \text{if} \quad x \in E; \\
      0 \quad \text{otherwise}.
    \end{cases}
    \label{eq:position_pvm}
\end{equation}
We can view $P_{\hat q}(E)$ as the formal version of a projector $\int_{E}\ket
  x\!\bra x \dd{x}$, with $\ket x$ a (non-normalisable) eigenvector of $\hat
q$ for the eigenvalue $x\in\mathbb R$. In other words, it is the projector onto
the subspace of states for which the position is contained in \(E\). In general,
\(P_{\hat{\bm x}}\) denotes the PVM associated to an arbitrary quadrature
\(\hat{\bm x}\).

Any linear map \(S\) on \(\R^{2M}\) such that \(\Omega(S\bm{x},S\bm{y})
  = \Omega(\bm{x},\bm{y})\) is called a \emph{symplectic transformation.} We use the following standard result:

\begin{lemma}
  \label{lemma:symplectic_representation}
 To any
  symplectic transformation $S$ we can associate a unitary operator \(\tau(S)\) such
  that
  \begin{equation}\label{mudef}
    \tau(S)^*\hat{\bm{x}}\tau(S) = \hat{\bm y},
  \end{equation}
  for all $\bm x\in\R^{2M}$, where $\bm y=S\bm x$.
  Furthermore, we have that
  \begin{equation}\label{Ptau}
    \hat D(\bm{y})^*P_{\hat{\bm{x}}}(E)\hat D(\bm{y}) = P_{\hat{\bm{x}}+\hat{\bm y}}(E) 
    \quad \text{and} \quad
    \tau^*(S)P_{\hat{\bm{x}}}(E) \tau(S) = P_{S\hat{\bm{x}}}(E).
  \end{equation}
\end{lemma}

\noindent See for example \cite{de_gosson_wigner_2017} for a proof.


\subsection*{Wigner functions}

\noindent There are several equivalent ways of defining the Wigner function of a quantum state
\cite{ferraro2005gaussian,cahill1969density, de_gosson_wigner_2017}. We follow
the conventions adopted in \cite{de_gosson_symplectic_2011} (see in particular
Prop.\ 175). The \emph{characteristic function} \(\Phi_\rho : \R^{2M} \to \C\) of
a density operator \(\rho\) on \(L^2(\R^{M})\) is defined as
\begin{equation}
    \Phi_\rho(\bm{x}) \defeq \Tr(\rho \hat D(-\bm{x})).
    \label{eq:characfunction}
\end{equation}
The \emph{Wigner function} \(W_\rho\) of \(\rho\) is then defined as the symplectic Fourier transform of the  characteristic function of $\rho$:
\begin{equation}
    W_\rho(\bm{x}) \defeq \operatorname{FT}[\Phi_\rho](J\bm{x}).
    \label{eq:WignerCharacFunction}
\end{equation}

\noindent The Wigner function is a real-valued square-integrable function on \(\R^{2M}\).
Arguably, its key property is that one can recover the probabilities for
quadrature measurements from its marginals \cite[proposition
6.43]{Gosson2006symplectic}. If \(W\) is the Wigner function of a pure state
\(\psi \in L^2(\R^{M})\) such that \(W\) is integrable on \(\R^{2M}\), then once
again identifying \(\bm{x}\) with \((q,p)\in\R^M\times\R^M\) in the same basis \((\bm
e_k, \bm f_k)_{k=1}^M\) as before,
\begin{align}
    \frac{1}{\sqrt{2 \pi}^M} \intg{\R^M}{W(q,p)}{p} &= |\psi(q)|^2, \\
    \frac{1}{\sqrt{2 \pi}^M}  \intg{\R^M}{W(q,p)}{ q} &= \left|\operatorname{FT}[\psi](p)\right|^2.
\end{align}
It follows that for any Lebesgue-measurable \(E \subseteq \R\), the probability
of obtaining an outcome $q$ in \(E\) when measuring the position $\hat q_1$ of the first qumode is given by
\begin{equation}
  \label{PW2}
  \mathrm{Prob}[q \in E | \rho]
  = \Tr(P_{\hat{q}_1}(E)\rho)
  = \frac{1}{(\sqrt{2\pi})^M} \smashoperator{\int_{E \times \R^{2M-1}}} \, {W_\rho(\bm{y})} \dd{\bm{y}}.\hspace{-1mm}
\end{equation}
In general, if \(\bm x \in \R^{2M}\) describes an arbitrary quadrature, one can show
that the probability of obtaining an outcome $x$ in \(E\subseteq \R\) when measuring the generalised quadrature $\hat{\bm x}$ is
\begin{equation}
  \label{eq:quadrature_probabilities}
  \mathrm{Prob}[x \in E | \rho]
  = \Tr(P_{\hat{\bm{x}}}(E)\rho)
  = \frac{1}{(\sqrt{2\pi})^M} \int_{A} {W_\rho(\bm{y})}\dd{\bm{y}},
\end{equation}
where \(A = \left\{\bm y \in \R^{2M} \mid \bm y \cdot \bm x \in E\right\}\). This corresponds to marginalising the Wigner function over the axes orthogonal to \(\bm{x}\).



Additionally, we make use of the symplectic covariance of the Wigner function:

\begin{lemma}
  \label{lemma:symplectic_covariance}
  Let \(\rho\) be a density operator on \(L^2(\R^M)\) and $S$ a symplectic
  transformation on \(\Xc\). Let \(\tau\) be as in
  Lemma~\ref{lemma:symplectic_representation}, then
  \(W_{\tau(S)\rho\tau^*(S)}(\bm x) = W_\rho(S\bm x)\).
\end{lemma}
\begin{proof}
  \begin{align}
    W_{\tau(S)\rho\tau^*(S)}(\bm x) &= \mathrm{TF}[\Phi_{\tau(S)\rho\tau^*(S)}](J\bm x) \\
                                    &= \frac{1}{(\sqrt{2\pi})^M}\intg{\Xc}{e^{-i J \bm x \cdot \bm y} \Tr(\tau(S) \rho \tau(S)^* D(-\bm y))}{\bm y} \\
                                    &= \frac{1}{(\sqrt{2\pi})^M}\intg{\Xc}{e^{-i J \bm x \cdot \bm y} \Tr(\rho \tau(S)^* D(-\bm y)\tau(S))}{\bm y} \\
                                    &= \frac{1}{(\sqrt{2\pi})^M}\intg{\Xc}{e^{-i J \bm x \cdot \bm y} \Tr(\rho D(-S\bm y))}{\bm y} \\
                                    &= \frac{1}{(\sqrt{2\pi})^M}\intg{\Xc}{e^{-i J \bm x \cdot (S^{-1}\bm y')} \Tr(\rho D(-\bm y'))}{\bm y'} \\
                                    &= \frac{1}{(\sqrt{2\pi})^M}\intg{\Xc}{e^{-i J (S\bm x) \cdot \bm y'} \Tr(\rho D(-\bm y'))}{\bm y'} \\
                                    &= \mathrm{TF}[\Phi_\rho](JS\bm x) \\
                                    &= W_\rho(S\bm x). \label{eq:WigSymplTrans}
  \end{align}
\end{proof}

\section{Continuous-variable contextuality}

\noindent In this section, we recall some definitions from the main text and briefly review the CV contextuality framework from~\cite{barbosa2019continuous}.

\subsection{Measurement scenarios} 
\label{subsec:measscenario}

\noindent An abstract description of an experimental setup is formalised as a \textit{measurement scenario}, which is a triple $\tuple{\Xc,\Mc,\bm \Oc}$:
\begin{itemize} 
\item 
in a given setup, experimenters can choose different measurements to perform on the physical system. Each possible measurement is labelled and $\Xc$ denotes the corresponding (possibly infinite) set of measurement labels;
\item
several compatible measurements can be implemented together (for instance, measurements on space-like separated systems). Maximal sets of compatible measurements define a \textit{context}, and $\Mc$ denotes the set of all such contexts;
\item
for a measurement labelled by $\bm x\in \Xc$, the corresponding outcome space is denoted $\bm \Oc_{\bm x} = \langle \Oc_{\bm x},\Fc_{\bm x} \rangle$, which has the structure of a measurable space with an underlying set $\Oc_{\bm x}$ and its associated $\sigma$-algebra $\Fc_{\bm x}$. The collection of all outcome spaces is denoted $\bm \Oc = (\bm \Oc_{\bm x})_{\bm x\in \Xc}$. For various compatible measurements labelled by elements of a set $U \subseteq X$, the corresponding joint outcome space is denoted $\bm\Oc_U\defeq\prod_{\bm x\in U}\bm \Oc_{\bm x}$. 
\end{itemize}
\noindent We refer the reader to \cite{barbosa2019continuous} for formal definitions.
Hereafter, we fix the measurement scenario $\tuple{\Xc,\Mc,\bm \Oc}_{\rm quad}$ as
follows:
\begin{itemize}
\item the set of measurement labels is $\Xc \defeq \R^{2M}$, the phase space
  with the symplectic structure described in the previous section;
\item the contexts are Lagrangian subspaces
  of $\R^{2M}$ so that the set of contexts $\Mc$ is the set of all Lagrangian subspaces of $\Xc$;
\item for each $\bm x \in \Xc$, the corresponding outcome space is $\bm \Oc_{\bm x} \defeq \langle \R,\Bc_\R
  \rangle$, so that for any set of measurement labels $U \subseteq \Xc$, $\Oc_U
  \cong \R^U$ can be seen as the set of functions from $U$ to $\R$ with its
  product $\sigma$-algebra $\Fc_U$.
  \footnotetext{The product \(\sigma\)-algebra is generated by the cylinders:
    subsets of \(\R^U\) of the form \(\prod_{\bm x \in U} E_{\bm x}\) such that
    \(E_{\bm x}\) is different from \(\R\) for at most a finite number of factors
    \(\bm x \in U\). This is the largest topology which makes the projections
    $\pi_{\bm x}$ measurable (the projections are given later in
    \eqref{eq:catprojection}).}
\end{itemize}

\noindent This measurement scenario is to be interpreted as follows. The measurement
corresponding to the label $\bm x \in \Xc$ is given by the measurement of the
corresponding quadrature $\hat{\bm{x}}$ (see Eq.~\eqref{eq:generalised_quadrature}), which itself is
formally described by the PVM $P_{\bm {\hat x}}$
(see Eq.~\eqref{eq:quadrature_probabilities}).

A pair of PVMs of self-adjoint operators is compatible, in the
sense that they admit a joint PVM, if and only if they commute \cite[Th.
9.19]{moretti_spectral_2017}, which in turn is true if and only if the operators
themselves commute \cite[Th. 9.41]{moretti_spectral_2017}.
Following Eq.~\eqref{eq:commutation_relation}, the PVMs associated to
$\bm x,\bm y\in \Xc$ commute if and only if $\Omega(\bm x,\bm y)=0$.
Thus, measurement labels are compatible only when they both belong to some
isotropic subspace of $\Xc$. The maximal isotropic subspaces are the Lagrangian
subspaces, so each context $L \in \Mc$ corresponds to a Lagrangian subspace.

For any set of measurement labels $U \subseteq \Xc$, an element \(f\) of \(\Oc_U\) is a function \(U \to \R\) and should be seen as
describing an experimental outcome obtained from measuring each quadrature in
\(U\): to each quadrature labelled by \(\bm{x} \in U\) is associated an outcome
\(f(\bm{x}) \in \R\). For each $\bm x \in U$, there is an associated
\emph{evaluation map}:
\begin{equation}
    \begin{aligned}
    \pi_{\bm x} : \Oc_U &\longrightarrow \R \\
    f &\longmapsto f(\bm x).
    \end{aligned}
    \label{eq:catprojection}
\end{equation}
If \(O \subseteq \Oc_U\) is a collection of such experimental outcomes,
\(\pi_{\bm x}(O) = \left\{ \pi_{\bm x}(f) = f(\bm x) \mid f \in O \right\}\) is
the implied collection of experimental outcomes for the quadrature \(\bm x\).

When the elements of \(U\) commute, \ie \(U\) is a subset of a context, such a function
describes a genuine experimental outcome that might be obtained from measuring
those quadratures in \(U\). It makes sense mathematically to extend this
definition to sets of non-commuting quadratures, disregarding the question of
compatibility of the associated measurements. Functions \(\Xc \to \R\) which
assign a tentative outcome to \emph{all} quadratures simultaneously are called
\emph{global value assignments}. In contrast, functions on contexts are called
\emph{local value assignments}.

\subsection{Empirical models}

\noindent Empirical models capture in a precise way the probabilistic \emph{behaviours}
that may arise upon performing measurements on physical systems. In practice, these amount to tables of normalised frequencies of outcomes gathered among various runs of the experiment, or to  tables of predicted outcome probabilities obtained by analytical calculation.
\begin{definition}[Empirical model] \label{def_empiricalmodel} 
An empirical model (or empirical behaviour) on a measurement scenario
$\tuple{\Xc,\Mc,\bm\Oc}$ is a family $e = \family{e_C}_{C \in \Mc}$, where $e_C$ is
a probability measure on the space $\bm \Oc_C$ for each context
$C \in \Mc$. 
\end{definition}

In general, empirical models should also satisfy the compatibility conditions: \(\forall C, C' \in \Mc, \quad e_C|_{C \cap C'} = e_{C'}|_{C \cap C'} \). 
This is a generalisation of the no-signaling condition to settings which might not have physical separation of observables. 
The compatibility condition ensures that the empirical behaviour of a given measurement or compatible subset of measurements is independent of which other compatible measurements might be performed along with them.
However, we only consider empirical models that arise from measurements on quantum systems, which therefore automatically verify these conditions.

\noindent In our case, for each context $L \in \Mc$, the set $\Oc_L = \prod_{\bm x \in L}
\R$ can be seen as the set of functions from $L$ to $\R$ with the corresponding product
$\sigma$-algebra.

We consider experiments \textit{arising from quadrature measurements of a quantum system}. 
We thus restrict our attention to empirical models $e = (e_L)_{L \in \Mc}$
satisfying the Born rule, which we refer to as quantum empirical models. In other words, there should exist some quantum state $\rho$
such that for all contexts $L \in \Mc$ and measurable sets $U \in \Fc_L$:
\begin{equation}
    e_L(U) = \Tr\left(\rho \prod_{\bm x \in L} P_{\bm{\hat{x}} } \circ \pi_{\bm x} (U) \right).
\end{equation}
We therefore use the notation $e^{\rho} = (e_L^\rho)_{L \in \Mc}$ to make explicit the dependence in $\rho$.
We may further and unambiguously write, for each $\bm x \in \Xc$ and for each $U
\in \Fc_{\bm x}$~\footnote{We write $\Fc_{\bm x}$ instead of $\Fc_{\left\{\bm x\right\}}$ for brevity.
}:
\begin{equation}
\label{eq:empiricalx}
    e^{\rho}_{\bm x}(U) = \Tr\left(\rho P_{\bm{\hat{x}} } \circ \pi_{\bm x} (U) \right).
\end{equation}

\subsection{Noncontextuality}

\noindent We begin by recalling the notion of hidden-variable model (HVM). Note that HVMs are also often referred to as ontological models \cite{Spekkens2005}. 

The idea behind the introduction of HVMs is that there may exist some space
$\bfLambda$ of hidden variables predetermining the empirical behaviour of a
physical system.
It is desirable to
impose constraints on HVMs which restrict the set of possible empirical behaviours and require that the system under consideration behaves \textit{classically} in some sense.
In the case of Bell-type scenarios, we require that the HVM must be local, \ie factorisable. To avoid fundamental indeterminacy, most hidden-variable theories impose a deterministic underlying description of the empirical behaviour.
However, hidden variables may not be directly accessible themselves, so we allow
that one may only have probabilistic information about the hidden variables in the form of a probability distribution $p$ over $\bfLambda$. 
The empirical behaviour should then be obtained as an average over the hidden-variable behaviours.

\begin{definition}
  A \emph{HVM} on a measurement scenario $\tuple{\Xc,\Mc,\bm\Oc}$ is a tuple
  $\tuple{\bm \Lambda, p, (k_C)_{C \in \Mc}}$ where:
  \begin{itemize}
  \item $\bfLambda = \tuple{\Lambda,\Fc_\Lambda}$ is the measurable space of hidden variables, 
  \item $p$ is a probability distribution on $\bfLambda$,
  \item for each context $C \in \Mc$, $k_C$ is a probability kernel between the measurable spaces $\bfLambda$ and $\bm \Oc_C$. It is measurable with respect to $\bfLambda$ and a probability measure over $\bm \Oc_C$.
  \end{itemize} 
\end{definition}

\noindent Determinism for the HVM is further ensured by requiring that each hidden variable gives a predetermined outcome. That is, for all contexts $C \in \Mc$ and for every $\lambda \in \Lambda$, $k_C(\lambda,\dummy) = \delta_{\bm x}$ is a Dirac measure at some $\bm x \in \Oc_C$.

\begin{definition}
  An empirical model \(e\) on a measurement scenario $\tuple{\Xc,\Mc,\bm\Oc}$ is
  \emph{noncontextual} if there exists a deterministic HVM $\tuple{\bm \Lambda, p, (k_C)_{C \in \Mc}}$ that reproduces the correct statistics by averaging over hidden variables, \ie
  \begin{equation}
    \label{eq:eval_HVM_emp}
    e_C(B) = \intg{\Lambda}{k_C(\lambda,B)}{p(\lambda)}.
  \end{equation}
  If \(e^\rho\) is noncontextual in the above measurement scenario $\tuple{\Xc,\Mc,\bm \Oc}_{\rm quad}$, we say that the density operator \(\rho\) is
  noncontextual for quadrature measurements.
\end{definition}

\noindent Crucially, it was shown in \cite{barbosa2019continuous} that there is a
canonical form of HVM for a noncontextual empirical model $e$
where the hidden-variable space can be taken as the space of global value assignments.
\begin{proposition}
\label{proposition:density}
For a noncontextual empirical model $e$ on a measurement scenario
$\tuple{\Xc,\Mc,\bm \Oc}$, there is a canonical HVM $\tuple{\bm \Lambda, p, (k_C)_{C \in \Mc}}$ for $e$
given by:
\begin{itemize}
    \item $\bfLambda \defeq \bm \Oc_{\Xc}$ (which is itself equal to \(\R^\Xc\),
      the set of functions \(\Xc \to \R\));
    \item $p$ is a probability measure on \(\bfLambda\);~\footnote{We modified slightly the definition of noncontextuality given in \cite{barbosa2019continuous} so it is easier to follow for readers unfamiliar with (non)contextuality. In particular, noncontextuality is defined there as the existence of a global probability distribution $d$ on $\bm \Oc_\Xc$ which gives back the correct statistics via marginalisation over a contexts. Then $p$ here would be taken as $p \defeq d$ from the hypothesis of noncontextuality.}
    \item for all global value assignments $\bm g \in  \Oc_\Xc$, all
      contexts $C \in \Mc$, and all measurables \(E \subseteq \R^C\),
      \begin{equation}
        k_C(\bm g,E) \defeq
        \begin{cases}
          1 \quad \text{if} \quad \bm g|_C \in E; \\
          0 \quad \text{otherwise}.
        \end{cases} \qedhere
      \end{equation}
\end{itemize}
\end{proposition}

\noindent In other words, without loss of generality, we can take the hidden variables to be the global value
assignments themselves, and the HVM makes experimental predictions as a
probability distribution over these global assignments. 

\section{``Measuring'' a displacement operator}
\label{app:measDisp}

\noindent In this section, we make the link with previous discrete-variable analogues of our result \cite{howard2014contextuality,delfosse2017equivalence}. In particular, in \cite{delfosse2017equivalence} it is not (clearly) specified that contextuality emerges with respect to the measurements of all (discrete) displacement operators. Since it might not be obvious what ``measuring a displacement'' refers to, as these are not self-adjoint operators, we detail below how it amounts to performing quadrature measurements.

We focus on a single qumode since it can be straightforwardly extended to multimode quantum states.
A quadrature operator such as the position operator $\hat q$ is self-adjoint and can be expanded via the spectral theorem \cite{hall2013quantum} as:
\begin{equation}
    \hat q = \intg{x \in \mathrm{sp}(\hat q)}{x}{P_{\hat q}(x)} \, ,
\end{equation}
where the spectrum of $\hat q$ is $\mathrm{sp}(\hat q) = \R$ and $P_{\hat q}$ is the PVM of $\hat q$ \cite[Th. 8.10]{hall2013quantum}.
For $E \in \Bc(\R)$, $P_{\hat q}(E)$ is given by \cite[Def. 8.8]{hall2013quantum}:
\begin{equation}
    P_{\hat q}(E) = \bm 1_E (\hat q).
\end{equation}
Informally, it assigns 1 whenever measuring $\hat q$ yields an outcome that belongs to $E$.
As previously mentioned, we can view $P_{\hat q}(E)$ as the formal version of the projector $\intg{x \in E}{\ket x\!\bra x}{x}$ (with $\ket x$ a non-normalisable eigenvector of $\hat q$).
Its functional calculus \cite{hall2013quantum} can be expressed as:
\begin{equation}
     f(\hat q) = \intg{x \in \mathrm{sp}(\hat q)}{f(x)}{P_{\hat q}(x)} \, ,
\end{equation}
for $f$ be a bounded measurable function.
Then, we can write the PVM of $f(\hat q)$ via the push-forward operation: 
\begin{equation}
    \forall E \in \Bc(\R), \; P_{f(\hat q)}(E) = P_{\hat q}(f^{-1}(E)).
\end{equation}
It follows immediately that, for $s \in \R$, the PVM for the diagonal phase operator $e^{is\hat q}$ is given, for $E \in \mathcal{B}(\mathbb{S}_1)$, by
\begin{equation}
    P_{\exp(is \hat q)}(E) \defeq P_Q(\{x \in \R \mid e^{isx} \in E\}).
\end{equation}

Define the rotated quadrature $\hat q_\theta \defeq \cos(\theta) \hat q + \sin(\theta) \hat p$ for $\theta \in \left[0,2\pi\right]$ and the phase-shift operator $\hat R(\theta)
\defeq \exp(i \frac{\theta}{2} (\hat q^2 + \hat p^2))$. Then
\begin{equation}
    \hat q_\theta = \hat R(\theta) \hat q \hat R(-\theta),
\end{equation}
so that the PVM of $\hat q_{\theta}$ is given by
\begin{equation}
    P_{\hat q_\theta}(E) = \hat R(\theta) P_{\hat q}(E) \hat R(-\theta).
\end{equation}
Let $(q,p) \in \R^2$. We can find $r \in \R_+$, $\theta \in \left[ 0,2\pi \right]$ such that $(q,p) = (-r\sin(\theta),r\cos(\theta))$.
Then:
\begin{equation}
    \hat D(q,p) = e^{i(p \hat q - q \hat p)} = e^{ir(\cos(\theta) \hat q - \sin(\theta) \hat p)} = e^{ir \hat q_{\theta}}.
\end{equation}
This form allows us to deduce the PVMs of displacement operators. For any $E \in \mathcal{B}(\mathbb{S}_1)$, we have
\begin{align}
    P_{\hat D(q,p)}(E) &= P_{\exp(ir \hat q_{\theta})}(E) \\
      &= P_{\hat q_\theta}(\{x \in \R \mid e^{irx} \in E\})\\
      &= \hat R(\theta) P_{\hat q}(\{x \in \R \mid e^{irx} \in E\}) \hat R(-\theta).
\end{align}

In conclusion, in the precise sense detailed above, ``measuring a displacement operator'' amounts to a quadrature measurement.

\end{document}